\documentclass[
reprint,
amsmath,amssymb,
aps,
pra,
]{revtex4-2}

\usepackage{graphicx}
\usepackage{amsmath,tikz,pgf,pgfplots,inputenc}
\usepackage{dcolumn}
\usepackage{subfigure}
\usepackage{bm}
\usepackage{hyperref}
\usepackage{color,soul}
\usepackage{xcolor}
\definecolor{blue1}{RGB}{0, 0, 255}
\usepackage{bbold}
\usepackage{appendix}
\usepackage{amsthm}

\usepackage{hyperref}
\usepackage{enumerate}

\newtheorem{theorem}{Theorem}%[section]
\newtheorem{lemma}[theorem]{Lemma}
\newtheorem{proposition}{Proposition}
\newtheorem{corollary}{Corollary}

\newenvironment{definition}[1][Definition]{\begin{trivlist}
\item[\hskip \labelsep {\bfseries #1}]}{\end{trivlist}}
\newenvironment{example}[1][Example.]{\begin{trivlist}
\item[\hskip \labelsep {\bfseries #1}]}{\end{trivlist}}
\newcommand{\ket}[1]{|#1\rangle}
\newcommand{\bra}[1]{\langle #1|}

\newcommand{\Tr}{{\mathrm {Tr}}}

\newcommand{\etal}{\textit {et al. }}

\begin{document}

\renewcommand*{\bibname}{Refrences}

%\preprint{APS/123-QED}

\title{Geometric Bloch Vector Solution to Minimum\textcolor{blue}{-}Error Discriminations of Mixed Qubit States}

\author{Mahdi Rouhbakhsh N.$^{1,2}$}
%\email{mahdi.rouhbakhshnabati@uni-tuebingen.de}

\author{and Seyed Arash Ghoreishi$^{3,4}$}
\email{arash.ghoreishi@savba.sk}

\affiliation{$^1$Sharif University of Technology, Department of Physics, 14588 Tehran, Iran}

\affiliation{$^2$Eberhard-Karls-Universität Tübingen, Institut für Theoretische Physik, 72076 Tübingen, Germany}

\affiliation{$^3$RCQI, Institute of Physics, Slovak Academy of Sciences, Dúbravská Cesta 9, 84511 Bratislava, Slovakia}

\affiliation{$^4$Faculty of Informatics, Masaryk University, Botanická 68a, 602 00 Brno, Czech Republic}

\begin{abstract}
We investigate minimum-error (ME) discrimination for mixed qubit states using a geometric approach. By analyzing positive operator-valued measure (POVM) solutions and introducing Lagrange operator $\Gamma$, we develop a four-step structured instruction to find $\Gamma$ for $N$ mixed qubit states. Our method covers optimal solutions for two, three, and four mixed qubit states, including a novel result for four qubit states. We introduce geometric-based POVM classes and non-decomposable subsets for constructing optimal solutions, enabling us to find all possible answers for the general problem of minimum-error discrimination for $N$ mixed qubit states with arbitrary a priori probabilities.
\end{abstract}

%\pacs{Valid PACS appear here}

\maketitle

\section{Introduction} \label{s1}

Distinguishing between quantum objects is one of the most fundamental tasks in information theory. Because of its particular importance in quantum communication~\cite{HolevoBook1982,HelstromBook1976} and quantum cryptography~\cite{GisinRMP2002}, discriminating among these objects plays a crucial role in quantum information protocols.

The quantum state discrimination problem has been the subject of many research during recent decades~\cite{BaeJPA2015}. The starting point of these efforts dates back to 1978 when Helstrom addressed this issue in his famous book~\cite{HelstromBook1976}. He studied the two-state case and obtained the minimum-error probability in the framework of quantum detection and estimation theory~\cite{BergouLNP2004}. Since then, based on the different approaches on the issue (error-free protocols with a possibility of failure and protocols with minimum-error probability) many developments have been achieved ~\cite{IvanovicPLA1987,PeresPLA1988,DieksPLA1988,JaegerPLA1995,CheflesPLA1998,CheflesBarnettPLA1998,Eldar2004,CrokePRL2006,HayashiPRA2008, SugimotoPRA2009,GhoreishiQIP2019,GhoreishiPRA2021,HaQIP2022,LoubenetsPRA2022}. Determining which one is the most suitable is highly dependent on the purpose of the process.

The extension to the general case of $N$ possible states $\{\rho_i\}_{i=1}^N$ with associated a priori probabilities $\{p_i\}_{i=1}^N$ is not a straightforward task. However, for these problems, there are necessary and sufficient conditions on the optimal measurement operators $\pi_i$ known as the Helstrom conditions~\cite{HelstromBook1976,BarnettAOP2009,BarnettJPA2009}
\begin{eqnarray} \label{HC1}
&\Gamma& -\; \: p_i \rho_i \geq \mathbb{0}    \: \:\:\:\: \:\:\:\:\:\:  \forall{i} ,  \\ \label{HC2}
&\pi_j& (p_j \rho_j - p_i \rho_i)\;\:\pi_i=\mathbb{0}   \:  \:  \: \forall{i,j}  ,
\end{eqnarray}
Here, $\Gamma=\sum_{i} p_i \rho_i \pi_i$ is a positive Hermitian operator known as the Lagrange operator. To prove sufficiency, let us consider ${\pi_j}$ as the optimal measurement operators. For any other POVMs ${\pi'_j}$, we have $\sum_{j} p_j \Tr(\rho_j \pi_j) \geq \sum_{i} p_i \Tr(\rho_i \pi'_i)$. By using the resolution of the identity for ${\pi'_j}$, i.e., $\sum_j\pi_j^{\prime}=\mathbb{1}$, we obtain $\sum_j \Tr\left(\left(\sum_{i} p_i \rho_i \pi_i-p_j\rho_j\right)\pi'_j\right) \geq 0$. Since the $\pi'_j$'s are positive elements, we can derive Eq. \eqref{HC1} from this inequality. The second condition, Eq. \eqref{HC2}, is not independent but rather can be derived from the first condition because $\Gamma$ is a positive and therefore Hermitian operator \cite{BarnettAOP2009}. The proof of necessity is straightforward \cite{BarnettJPA2009}. So, the main condition, which is both a necessary and sufficient condition for an optimal measurement, is Eq.~\eqref{HC1}. However, we will still use Eq.~\eqref{HC2} in our procedure. It is worth emphasizing that all minimum-error measurements that give the optimal probability define a unique Lagrange operator \cite{BarnettBook2009}.

The Helstrom conditions can be used to verify the optimality of a candidate measurement. However, they cannot be directly employed to construct the optimal measurement for a general problem. They can be helpful for problems with certain symmetries among states, providing guidance in guessing the appropriate form of measurement in such problems~\cite{BanIJTP1997,BarnettPRA2001,ChouPRA2003,AnderssonPRA2002,ChouPRA2004,MochonPRA2006}.

Although, for a while, only problems with certain symmetries seemed solvable, in recent years there have been successful attempts in solving problems by employing the conditions~\eqref{HC1} and~\eqref{HC2}. Particularly, for qubit states, the Bloch representation provides a useful tool to solve the problem of ME discrimination of qubit states. For example, Hunter in 2004 provided a complete solution for pure qubit states with equal a priori probabilities \cite{HunterAIP2004}. Samsonov in 2009 employed the necessary and sufficient conditions to propose an algorithmic solution for $N$ pure qubit states~\cite{SamsonovPRA2009}. Deconinck and Terhal in 2010 used the discrimination duality theorem and the Bloch sphere representation to provide a geometric and analytic representation of the optimal guessing strategies for qubit states~\cite{DeconinckPRA2010}. Bae in 2013 proposed a geometric formulation for a case with equal priori probabilities in a situation where quantum state geometry is clear~\cite{BaeNJP2013,BaePRA2013}, and in the same year, a complete analysis for three mixed states of qubit systems was done by Ha \etal~\cite{HaPRA2013}. Weir \etal in 2017 utilized the Helstrom conditions constructively and analytically to solve a problem with any number of qubit states with arbitrary priori probabilities~\cite{WeirPRA2017}, giving the central role in their approach to the inverse of $\Gamma$ instead of $\Gamma$ itself. Later, they applied their method to find optimal strategies for trine states~\cite{WeirQST2018}.

In this paper, we consider the problem of ME discrimination of $N$ mixed qubit states with arbitrary a priori probabilities using a different geometric approach. 
Our starting point is the Helstrom conditions. By employing both these conditions and the Bloch representation for qubit states, one can find all possible ME POVM answers for a given problem by knowing only the Lagrange operator $\Gamma$. We provide a comprehensive four-step structured instruction for obtaining $\Gamma$ and solving the problem for arbitrary prior probabilities.
We will show that $\Gamma$ can be found by considering a maximum of four qubits of the problem. Therefore, our instruction can be used to find $\Gamma$, for a general minimum-error problem, with one, two, three, and four qubits in the corresponding steps. A novel and direct solution to the last case of four-qubit states is presented. Equipped with these tools, we can find all possible ME discrimination measurements for the problem. Moreover, by introducing non-decomposable POVM answers, an alternative approach to achieve a general optimal answer is possible by considering a convex combination of these non-decomposable POVM sets.

The paper is organized as follows: In Sect.~\ref{section:Quantum state discrimination}, we provide an overview of the problem of quantum state discrimination using a general approach. We define the Primal and Dual problem formulations. In Sect. ~\ref{subsection:Discrimination of qubit states}, we adopt the Bloch vector representation and present a systematic method to determine all optimal solutions to the problem using the Lagrange operator $\Gamma$. We construct a four-step instruction that guides this process effectively.Next, in Sect.~\ref{section:Solution in qubit state discrimination}, we solve a general problem of mixed qubit states with arbitrary priori probabilities for two, three, and four qubits cases. We begin Sec~\ref{Ncases} by introducing an important theorem for the problem of $N$ qubit states discrimination for $N>4$ Then, in Sect.~\ref{section:examples}, we consider cases with $N>4$, such as $N=5$ and $N=6$ states, as illustrative examples. In this analysis, we identify regions by the simplest non-decomposable solutions, which are represented by different colors. We also revisit the special case of $N$ states with equal a priori probabilities. The last section includes a summary and conclusions of the paper.\\

\section{Quantum state discrimination: Primal vs Dual problem} \label{section:Quantum state discrimination}

Suppose that we are given $N$ states $\rho_i$, $i=1, \cdots, N$, with a priori probabilities $p_i$. Our task is to discriminate among these states by performing the optimal measurement with minimum probability of error. Such measurement needs to have $N$ outcomes corresponding, respectively, to each of the  $N$ states $\rho_i$. A general measurement can be described by a positive operator-valued measure (POVM)~\cite{PeresBook1993,NielsenBook2010,TeikoBook2011}, which is defined by a set of operators $\{ \pi_j \}$ satisfying
\begin{eqnarray} \label{POVMCompleteness}
\pi_j \geq \mathbb{0}, \qquad
\sum_{j}\pi_{j}=\mathbb{1}.
\end{eqnarray}
Each of the possible measurement outcomes $j$ is characterized by the corresponding POVM element $\pi_j$. The probability of observing outcome ``$j$" when the state of system is ``$i$" is
\begin{equation}
P(j|i)=\Tr(\pi_j\rho_i).
\end{equation}
Then, the probability of making a correct guess when the given set is $\{p_i,\rho_i\}_{i=1}^N$, will be
\begin{equation}
P_{\textrm{corr}}=\sum_{i=1}^N p_i\Tr(\rho_i\pi_i).
\end{equation}
In the minimum-error discrimination approach, the goal is to find the optimal measurement that minimizes the average probability of error, or equivalently, maximizes the probability of making a correct guess denoted by $P_{\textrm{guess}}=\max \{P_{\textrm{corr}}\}$. This can be achievable, as we mentioned previously, if and only if the POVM elements satisfy the Helstrom conditions~\eqref{HC1} and~\eqref{HC2}. This problem is called \emph{primal} problem. Equivalently, the \emph{dual} problem can be defined using semidefinite programming as follows 
\begin{equation}
    \min \Tr{[K]},\qquad \text{subject to}  \qquad K-p_i\rho_i\geq 0, 
\end{equation}
the constraints on the dual problem can be written as 
\begin{equation}
    K=p_ir_i+r_iw_i,
\end{equation}
where in this equation, we introduce the variable $K$ as a positive semidefinite matrix and non-negative variables $r_i$, where $i=1, \cdots, N$. Additionally, $w_i$ represents density matrices. \\
It is of importance that in general these two approaches may not coincide. The condition under the guessing probability can be obtained from primal and dual optimizations following the property called strong duality. The conditions under which the dual problem coincides with the primal one are Karush-Kuhn-Tucker (KKT) \cite{BoydBook2004}
conditions and can be written as 
\begin{eqnarray}
    r_iw_i-r_jw_j&=&p_j\rho_j-p_i\rho_i, \\
    r_i\Tr[w_i\pi_i]&=&0. \nonumber
\end{eqnarray}
So far, most of the works employ the duality problem to tackle the problem of minimum-error discrimination of qubit states \cite{DeconinckPRA2010,BaeNJP2013,BaePRA2013,HaPRA2013}. However, in this paper, we will start directly from the primal problem and use the Helstrom conditions~\eqref{HC1} and~\eqref{HC2} in a constructive way for solving the problem.

To continue, let us reformulate conditions~\eqref{HC1} and~\eqref{HC2} by summing the second one over $j$ which results in
\begin{eqnarray}
\Gamma - \tilde{\rho}_i\geq \mathbb{0}  \enspace\enspace\enspace\enspace\enspace \forall{i} ,\label{gamma}
\\
(\Gamma - \tilde{\rho}_i)\pi_i=\mathbb{0}  \enspace\enspace\enspace \forall{i} ,\label{pi0}
\end{eqnarray}
where we have defined  $\tilde{\rho}_i=p_i\rho_i$.

From Eqs. ~\eqref{gamma} and \eqref{pi0} (and its Hermitian conjugate $\pi_i(\Gamma - \tilde{\rho}_i)=\mathbb{0}  \enspace\enspace \forall{i}$), in the case where $\Gamma - \tilde{\rho}_i$ is not a full rank operator, it can be easily seen that $\textrm{supp}\{\Gamma - \tilde{\rho}_i\}\perp \textrm{supp}\{\pi_i\}$, or equivalently, $\textrm{supp}\{\Gamma - \tilde{\rho}_i\}\subseteq  \textrm{ker}\{\pi_i\}$, where $\textrm{supp}\{\cdot\}$ and $\textrm{ker}\{\cdot\}$ denote the support and kernel of an operator, respectively. In other words, Eqs ~\eqref{gamma} and \eqref{pi0} imply that both determinants $\det\{\Gamma - \tilde{\rho}_i\}$ and $\det\{\pi_i\}$ have to be null.
\\
Thus far, the discussion has remained general without any specific reference to qubit states. From this point onward, we will focus on qubit states for the rest of this paper.

\section{Minimum-error discrimination of qubit states} \label{subsection:Discrimination of qubit states}

Solving a general problem of minimum-error (ME) discrimination is a challenging task, and there is no analytical solution in general. However, there have been some approaches proposed for qubit states in the literature~\cite{HunterAIP2004,SamsonovPRA2009,DeconinckPRA2010,BaeNJP2013,HaPRA2013,WeirPRA2017}. In this work, we consider the problem of the most general case of qubit states using a geometric approach. We utilize the Bloch vector representation to express $\rho_i$ and  $\Gamma$ as follows
\begin{eqnarray}
\rho_i &=&\frac{1}{2}(\mathbb{1}+\boldsymbol{v}_i \cdot \boldsymbol{\sigma}), \label{eq:defrho} \\
\Gamma &=&\frac{1}{2}(\gamma_0\mathbb{1}+\boldsymbol{\gamma} \cdot \boldsymbol{\sigma}). \label{eq:defGamma}
\end{eqnarray}
Here, $\boldsymbol{v}_i=(v_{xi},v_{yi},v_{zi})$ represents the Bloch vector corresponding to the state $\rho_i$, $\boldsymbol{\sigma}=(\sigma_x, \sigma_y, \sigma_z)$ is a vector constructed by Pauli matrices. Additionally, $\gamma_0=\Tr(\Gamma)=P_{\textrm{guess}}$ and  $\boldsymbol{\gamma}=(\gamma_x, \gamma_y, \gamma_z)$.
Non-negativity of $\rho_i$ implies that $0\le |\boldsymbol{v}_i|\le 1$. The lower bound is achieved for maximally mixed states, while the upper bound corresponds to pure states. Therefore, $|\boldsymbol{v}_i|$ can be considered as a measure of purity.\\
By using Eqs.~\eqref{eq:defrho} and \eqref{eq:defGamma}, we can write
\begin{equation} \label{eq:Gamma-rho}
\Gamma - \tilde{\rho}_i=\frac{1}{2}[(\gamma_0-p_i)\mathbb{1}+(\boldsymbol{\gamma}-\boldsymbol{\tilde{v}_i})\cdot\boldsymbol{\sigma}],
\end{equation}
where we have defined the sub-normalized Bloch vector $\boldsymbol{\tilde{v}_i}=p_i\boldsymbol{v}_i$ associated with the state $\rho_i$ and its priori probability $p_i$. However, for simplicity, both vectors $\boldsymbol{v}_i$ and $\boldsymbol{\tilde{v}_i}$ are referred to as Bloch vectors in this paper. It follows from inequality~\eqref{gamma} that
\begin{equation} \label{eq:gamma0-pgeq}
\gamma_0-p_i\geq|\boldsymbol{\tilde{v}_i}-\boldsymbol{\gamma}|.
\end{equation}
This equation is equivalent to the main Helstrom condition, Eq.~\eqref{HC1}. From now on, we will check the optimality of our solutions by Eq.~\eqref{eq:gamma0-pgeq}.\\ Moving forward, if the inequality in Eq.~\eqref{eq:gamma0-pgeq} holds strictly, indicating positive definiteness of the determinant of $\Gamma-\tilde{\rho}_i$, it implies that the corresponding measurement operator $\pi_i$ must be null. However, if the equality holds, indicating a zero determinant of $\Gamma-\tilde{\rho}_i$, it does not necessarily mean that $\pi_i$ has to be null. Therefore, in the case of a nontrivial answer, we obtain the following useful formula

\begin{equation} \label{eq:gamma0-p}
\gamma_0-p_i=|\boldsymbol{\tilde{v}_i}-\boldsymbol{\gamma}|.
\end{equation}
Regarding $\gamma_0=P_{\textrm{guess}}$, this relation indicates that the distance between two vectors $\boldsymbol{\gamma}$ and $\boldsymbol{\tilde{v}_i}$ is equal to the difference between two probabilities: the guessing probability and the prior probability $p_i$ associated with state $\rho_i$. The greater the value of $p_i$, the closer the vectors $\boldsymbol{\tilde{v}_i}$ and $\boldsymbol{\gamma}$ are to each other. Keeping in mind that Eq.~\eqref{eq:gamma0-p} does not hold for all states in general. For simplicity, and without loss of generality, we arrange the set of states $\{\rho_i\}_{i=1}^N$ and their corresponding priori probabilities $\{p_i\}_{i=1}^N$ in such a 
way that for the first $M$ states $\{\rho_i\}_{i=1}^M$, the operator $(\Gamma - \tilde{\rho}_i)$ has a rank of less than 2. This implies that these states are detectable in the optimal case. On the other hand, the remaining states $\{\rho_i\}_{i=M+1}^N$ have  $(\Gamma - \tilde{\rho}_i)$ as a full rank operator, meaning that there is no way to detect them optimally. So, for $i=M+1,\cdots,N$, we have $\pi_i=\mathbb{0}$.
With this terminology, and because $\det(\pi_i)=0$ when $\Gamma - \tilde{\rho}_i\ne \mathbb{0}$ or when it is a full rank operator, it turns out that the measurement operators $\pi_i$'s for $i=1,\cdots,M$ must be rank one and proportional to projectors. As a qubit projector is characterized by its unit vector $\boldsymbol{\hat{n}_i}$ on the Bloch sphere, we can write $\pi_i$ as follows
\begin{equation} \label{eq:defpi}
\pi_i=\frac{\alpha_i}{2}(\mathbb{1}+\boldsymbol{\hat{n}_i} \cdot \boldsymbol{\sigma}),
\end{equation}
for $i=1,\cdots,M$, where $\alpha_i$'s are real numbers ranging from zero to one, and $\alpha_i=\Tr(\pi_i)$. The completeness relation \eqref{POVMCompleteness} requires the following conditions on the parameters $\alpha_i$'s
\begin{eqnarray}{}
\sum\limits_{i=1}^{M}\alpha_i=2, \qquad
\sum\limits_{i=1}^{M}\alpha_i\boldsymbol{\hat{n}_i}=\boldsymbol{0},
\label{eq:sigmaalpha}
\end{eqnarray}
where $0 \leq \alpha_i\leq 1$. On the other hand, by inserting Eqs.~\eqref{eq:Gamma-rho} and \eqref{eq:defpi} into the Helstrom condition~\eqref{pi0}, and using Eq.~\eqref{eq:gamma0-p}, we find the unit vector $\boldsymbol{\hat{n}_i}$ as follows
\begin{equation} \label{eq:ni}
\boldsymbol{\hat{n}_i}=\frac{\boldsymbol{\tilde{v}_i}-\boldsymbol{\gamma}}{|\boldsymbol{\tilde{v}_i}-\boldsymbol{\gamma}|} .
\end{equation}
Obviously, knowing the vector $\boldsymbol{\gamma}$ is equivalent to knowing all the unit vectors $\boldsymbol{\hat{n}_i}$ corresponding to the non-null measurements $\pi_i$.
\begin{corollary} \label{corollary:1}
Based on Eqs.~\eqref{eq:sigmaalpha}, $\boldsymbol{\gamma}$ has to be confined within the convex polytope of the points $\boldsymbol{\tilde{v}_i}$'s.
\end{corollary}

\begin{corollary} \label{corollary:2}
From Eqs.~\eqref{eq:gamma0-p} to \eqref{eq:sigmaalpha}, it can be deduced that by translating all the vectors $\boldsymbol{\tilde{v}_i}$'s, while keeping $p_i$ unchanged, by a fixed vector inside the Bloch sphere, the POVM answers will be unchanged. The relative positions of $\boldsymbol{\tilde{v}_i}$'s are the only relevant factor. Therefore, if all $\boldsymbol{\tilde{v}_i}$'s move with the same displacement vector $\boldsymbol{a}$ ($\boldsymbol{\tilde{v}_i}\to\boldsymbol{\tilde{v}_i}+\boldsymbol{a}$), we expect the same answer. In other words, a set of $\{ p_i , \rho'_i \}$ yields the same measurement operators and guessing probability as the initial set $\{ p_i , \rho_i \}$, where
\begin{equation}
\rho'_i=\frac{1}{2}[\mathbb{1}+(\boldsymbol{v}_i+\frac{\boldsymbol{a}}{p_i}).\boldsymbol{\sigma}]
\quad\quad
\text{with}
\quad\quad
|\boldsymbol{v}_i+\frac{\boldsymbol{a}}{p_i}|\leq 1
\label{eq:rhoprimesameasrho}
\end{equation}
This class is characterized by an unchanged guessing probability.
\end{corollary}

From the preceding discussion, it is evident that to characterize the optimal measurements, both the vector $\boldsymbol{\gamma}$ and the real parameters $\alpha_i$'s need to be known.
Once $\boldsymbol{\gamma}$ is known, we can find all possible POVM answers using the following theorem.

\begin{theorem}
Assume that we have the operator $\Gamma$ \textnormal{(}or equivalently $\boldsymbol{\gamma}$\textnormal{)} for the corresponding problem of $N$ arbitrary mixed qubit states $\{p_i, \rho_i\}_{i=1}^N$. Then, finding all possible minimum-error POVM measurements of the problem is equivalent to finding all possible sets of $\{\alpha_i\}_{i=1}^{M}$ \textnormal{(}$0\leq\alpha_i\leq 1$\textnormal{)} that satisfy Eq. ~\eqref{eq:sigmaalpha}.
\end{theorem}

\begin{proof}
Since $\Gamma$ is given, we can determine $\gamma_0$ and $\boldsymbol{\gamma}$ using Eq.~\eqref{eq:defGamma}. By applying Eq.~\eqref{eq:gamma0-p}, we can identify the $M$ states that satisfy this equation. Then, utilizing Eq.~\eqref{eq:defpi} and the completeness relation for a POVM, we end up with Eq.~\eqref{eq:sigmaalpha}, where the corresponding unit vectors $\boldsymbol{\hat{n}_i}$ are defined by Eq.~\eqref{eq:ni}. Each set of $\{\alpha_i\}_{i=1}^{M}$ that satisfies Eq.~\eqref{eq:sigmaalpha} corresponds to a POVM set $\{\pi_i\}_{i=1}^{M}$ using Eq.~\eqref{eq:defpi}. From the explanation preceding Eq.~\eqref{eq:gamma0-p}, the $M$ operators fulfill the equality part of the Helstrom condition, Eq.~\eqref{eq:gamma0-pgeq}. For the remaining states in the problem, the corresponding operators must be null, satisfying the inequality part of Eq.~\eqref{eq:gamma0-pgeq}. Therefore, each set of $\{\alpha_i\}_{i=1}^{M}$ corresponds to an optimal POVM for the problem.
\end{proof}
On the other hand, if we are given a set of ME discrimination POVM $\{\pi_i\}$, we can find the Lagrange operator $\Gamma$ (or equivalently $\boldsymbol{\gamma}$) using the relation $\Gamma = \sum\limits_i\tilde{\rho}_i\pi_i$. Then, by utilizing Eq.~\eqref{eq:sigmaalpha}, we can obtain all possible optimal answers.
\subsection{Types of states and measurements in an optimal strategy}\label{subsection:Classification of optimal answers}
Similar to \cite{DeconinckPRA2010}, we distinguish the states in a general qubit state discrimination problem. We have shown that $\boldsymbol{\gamma}$ must be confined within the convex polytope of the points $\boldsymbol{\tilde{v}_i}$'s. Therefore, based on different representations of $\boldsymbol{\gamma}$, the states in a discrimination problem can be classified as unguessable, nearly guessable, and guessable states \cite{DeconinckPRA2010}. The \emph{unguessable} states are those for which $\Gamma - \tilde{\rho}_i$ is a full rank operator (inequality part of Eq.~\eqref{eq:gamma0-pgeq}) and their corresponding POVM elements are always null, indicating that they cannot be detected optimally. The \emph{nearly guessable} states are those for which the operator $\Gamma - \tilde{\rho}_i$ is not full rank, but they do not appear in any optimal measurement, and therefore, their related measurement operators are also null. Finally, the \emph{guessable} states are those that satisfy Eq.~\eqref{eq:gamma0-p}, and their POVM elements are non-null for some optimal measurements, indicating that they can be detected optimally.

It is important to study the conditions under which the optimal measurement of a discrimination problem is unique. First, based on the preceding discussion, we can state the following proposition for a discrimination problem with $N\leq4$
\begin{proposition}
For a discrimination problem with $N\leq4$, the set of parameters $\{ \alpha_i \}$ satisfying Eq.~\eqref{eq:sigmaalpha} is unique if the states form a $N-1$ simplex inside the Bloch sphere. However, this uniqueness condition does not necessarily hold for $N>4$ or $N\leq4$ when the states do not form a $N-1$ simplex. In such cases, each possible set of $\{\alpha_i\}$ can lead to a complete solution for the optimal measurement operators.
\end{proposition}
\begin{proof}
The proof can be derived from theorem 3.5.6 in \cite{LeonardBook2016}. According to this theorem, a point in a k-simplex with vertices $x_0,x_1,\cdots,x_k$ can be uniquely expressed as a convex combination of the vertices. So, in the problem of $N\leq4$ qubit state discrimination, if these states form a simplex, the convex combination for $\boldsymbol{\gamma}$ will be unique.
\end{proof}

Consequently, in cases where the solutions are non-unique, we may encounter situations where different POVM measurements yield the same optimal guessing probability. Thus, it is possible to construct different measurements by combining these measurements using convex combinations. This fact motivates us to categorize optimal measurements into two different families: \emph{decomposable} and \emph{non-decomposable} measurements.
\begin{definition}
Consider a POVM $M=\{\pi_1,\pi_2,\cdots ,\pi_m\}$ which is an optimal measurement for some discrimination problems. According to Eq.~\eqref{eq:defpi}, we can associate this POVM with a set of Bloch unit vectors $E=\{\boldsymbol{\hat{n}_1},\boldsymbol{\hat{n}_2},\cdots ,\pmb{\hat{n}_m}\}$, such that their convex polytope contains the Origin (Eq~\eqref{eq:sigmaalpha}). If it is possible to find a subset $E'\subset E$ such that the convex polytope formed by $E'$ still contains the Origin, then the POVM $M$ is considered \emph{decomposable}. Otherwise, if no such subset exists, it is categorized as \emph{non-decomposable}. Due to the three-dimensional nature of the directions represented by $\pmb{\hat{n}_i}$'s, it follows that \emph{non-decomposable sets can have at most four elements.}
\end{definition}
An exception for this definition is the existence of a measurement operator $\pi_j$ that is proportional to the unit matrix $\mathbb{1}$. In this case, there exists a one-element non-decomposable set $M'=\{\pi_i=\delta_{ij}\mathbb{1}\}$.\\ Now, we can derive an important result from the last definition, which justifies the inclusion of only four steps in our instruction: We can always find a non-decomposable solution for a given problem. \emph{Since each solution is characterized by the same operator $\Gamma$, it is always possible to find that operator using a maximum of four states of the problem which satisfy Eq.~\eqref{eq:gamma0-p}.}

\begin{theorem}\label{theorem:decomposable}
Let $M=\{\pi_i\}_{i=1}^m$ be a positive operator-valued measure (POVM) which is an optimal measurement for some minimum-error discrimination problems $\{ p_i, \rho_i \}_{i=1}^{N}$, with the corresponding set of measurement directions $E=\{\boldsymbol{\hat{n}_1},\boldsymbol{\hat{n}_2},\cdots ,\pmb{\hat{n}_m}\}$.If $M$ is decomposable, then there exists a non-empty subset $E'\subset E$ such that the corresponding POVM $M'$ is also optimal for the same minimum-error discrimination problem $\{ p_i, \rho_i \}_{i=1}^{N}$.
\end{theorem}
\begin{proof}
We have shown that a ME problem is characterized by an operator $\Gamma$, and for those states that satisfy Eq.~\eqref{eq:gamma0-p}, we can obtain $\alpha_i$ and $\boldsymbol{\hat{n}_i}$ from Eqs. \eqref{eq:sigmaalpha} and \eqref{eq:ni}. ME POVM measurements are only dependent on the set $\{\alpha_i\}$, and each possible set of $\{\alpha_i\}$ leads to a different POVM. If we have a set of $E=\{\alpha_i\}$ corresponding to the POVM $M$, then eliminating $\boldsymbol{\hat{n}_k}$ is equivalent to consider the corresponding $\alpha_k$, or equivalently $\pi_k$, to zero. From the decomposable definition, we suppose that the convex polytope of the reduced set $E'$ contains the origin, so Eq.~\eqref{eq:sigmaalpha} is satisfied. Therefore, there is a POVM $M'$ corresponding to $E'$ that is a ME solution to the problem as well.
\end{proof}
In simpler words, any subset $E'\subset E$ consisting of at most four elements of $\pmb{\hat{n}_i}$'s that contains the origin within the convex polytope of its elements, and cannot have any elements removed while still keeping the origin within the convex polytope of the remaining elements, is considered a \emph{non-decomposable subset}.

\subsection{Qubit states: Arbitrary priori probabilities} \label{subsection:Qubit states: Arbitrary priori probabilities}
Now, we will take a step further by considering a general problem of a set of qubit states with arbitrary prior probabilities. Before we provide an instruction on how to find answers for a general problem $\{ p_i, \rho_i \}_{i=1}^N$, it is important to note that we can find the Lagrange operator $\Gamma$ by using Eq.~\eqref{eq:gamma0-p}. Based on our result from the definition of non-decomposable POVMs, we know that a general minimum-error problem always has a non-decomposable solution, using a maximum of four states $\rho_i$'s that satisfy Eq.~\eqref{eq:gamma0-p}. This means that to find an answer, one could test a total of $\sum_{k=1}^4\binom{N}{k}$ possibilities! However, with the following instruction, we can find an answer much more quickly.

{\it Step  (i).---}First, we examine whether there exists a $\rho_j$ \emph{that has the greatest priori probability} ($p_j > p_i , \forall i$). If such a state $\rho_j$ exists, we need to check if the Lagrange operator $\Gamma$ is equal to $\tilde{\rho}_j$ or not. If $\Gamma=\tilde{\rho}_j$, it implies that $\gamma_0=p_j$ and $\boldsymbol{\gamma}=\boldsymbol{\tilde{v}_j}$ according to Eqs.~\eqref{eq:defrho} and \eqref{eq:defGamma}. Then, to satisfy the Helstrom condition given by Eq.~\eqref{eq:gamma0-pgeq}, the following condition must be hold
\begin{equation} \label{pj-pi}
p_{ji}\geq \tilde{d}_{ji}.
\end{equation}
where $p_{ji}:=p_j-p_i$ and $\tilde{d}_{ji}:=|\boldsymbol{\tilde{v}_j}-\boldsymbol{\tilde{v}_i}|$.\\If the above condition is satisfied for every $i$, the optimal POVM answer will be $\pi_i=\delta_{ij}\mathbb{1}$, and the guessing probability $P_{\textrm{guess}}$ will be equal to $p_j$. It is called ``no measurement strategy," where the state $\rho_j$ is simply guessed \cite{HunterPRA2003}.

{\it Step  (ii).---} If the above condition is not satisfied, we proceed to test pairs of qubit states. First, let us consider two qubit states $\rho_l$ and $\rho_m$ with arbitrary priori probabilities $p_l$ and $p_m$ ( $p_l\geq p_m$). By applying Eq.~\eqref{eq:gamma0-p} to states $\rho_l$ and $\rho_m$, and then subtracting the results, we obtain
\begin{equation} \label{eq:pl-pm}
p_{lm}= |\boldsymbol{\tilde{v}_m}-\boldsymbol{\gamma}| - |\boldsymbol{\tilde{v}_l}-\boldsymbol{\gamma}|,
\end{equation}
which defines a hyperbola with its foci at $\boldsymbol{\tilde{v}_l}$ and $\boldsymbol{\tilde{v}_m}$. This hyperbola is characterized by parameters $a$, $b$, and $c$, which can be derived from the Bloch vectors of states $\rho_l$ and $\rho_m$, and their corresponding probabilities $p_l$ and $p_m$ (see Fig.~\ref{fig:hyperbola}), as
\begin{equation} \label{eq:hyperbola1}
a=\frac{1}{2}~p_{lm}, \qquad
c=\frac{1}{2}~\tilde{d}_{lm},
\end{equation}
and
\begin{equation} \label{eq:hyperbola2}
b=\sqrt{c^2-a^2},
\end{equation}
It follows from Eq.~\eqref{eq:gamma0-pgeq} that only the left branch of the hyperbola could provide a candidate answer for $\boldsymbol{\gamma}$.

\begin{figure}
\mbox{\includegraphics[scale=0.4]{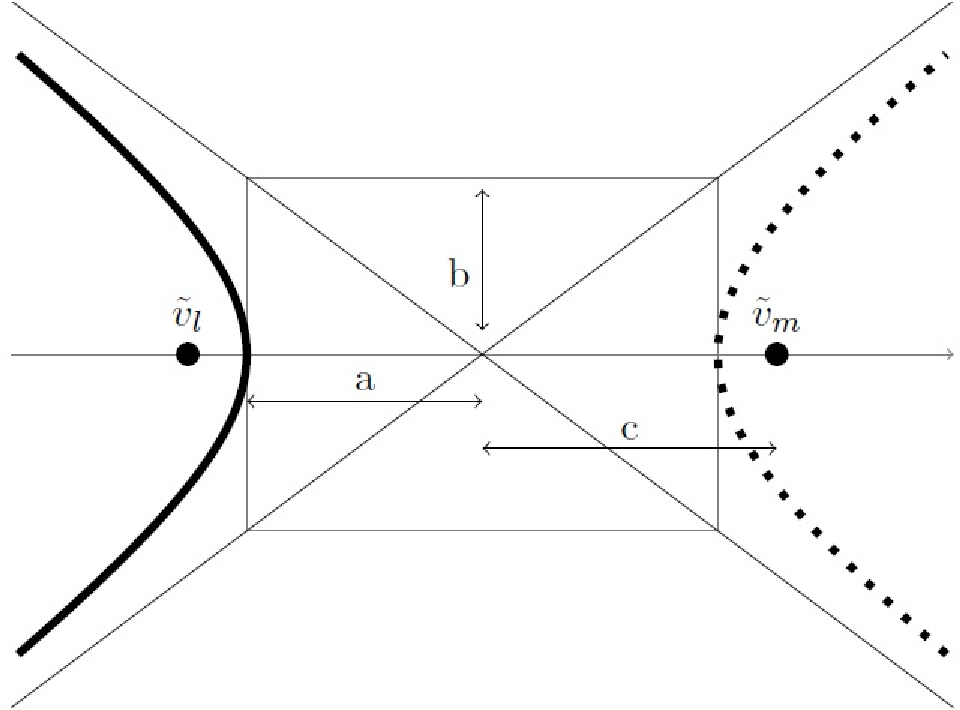}}
\caption{The hyperbola given by Eq.~\eqref{eq:pl-pm}, which only the left part of this hyperbola can be considered as an allowed area for $\boldsymbol{\gamma}$. (The existence of this hyperbola is dependent on the positivity of Eq.~\eqref{eq:Rlm})}
\label{fig:hyperbola}
\end{figure}

In this case, the candidate $\boldsymbol{\gamma}$ is located at the distance of
\begin{equation} \label{eq:Rlm}
\tilde{R}_{lm}=c-a=\frac{1}{2} \left[ \tilde{d}_{lm} - p_{lm} \right],
\end{equation}
from point $\boldsymbol{\tilde{v}_l}$ on the connecting line between $\boldsymbol{\tilde{v}_l}$ and $\boldsymbol{\tilde{v}_m}$. Since $\tilde{R}_{lm}$ is a distance, it must be positive. Therefore, possible candidates for $\gamma_0$ and $\boldsymbol{\gamma}$ are
\begin{equation}\label{eq:gamma0lm}
\gamma_{0_{lm}}'~=\frac{1}{2}\left[\left(p_l+p_m\right)+\tilde{d}_{lm}\right],
\end{equation}
and
\begin{eqnarray}\label{eq:gammalm}
\boldsymbol{\gamma}'_{lm}&=&\frac{1}{2} \left[ (\boldsymbol{\tilde{v}_l}+\boldsymbol{\tilde{v}_m})+\frac{p_{lm}}{\tilde{d}_{lm}}(\boldsymbol{\tilde{v}_l}-\boldsymbol{\tilde{v}_m}) \right] \\ \nonumber
&=&\frac{1}{2}\left[\left(1+\frac{p_{lm}}{\tilde{d}_{lm}}\right)\boldsymbol{\tilde{v}_l}+\left(1-\frac{p_{lm}}{\tilde{d}_{lm}}\right)\boldsymbol{\tilde{v}_m}\right],
\end{eqnarray}
where the primes indicate that they are still candidate answers. Now, the question is: which two states to pick? In this case, we have the following proposition
\begin{proposition}
In a case where a two-element optimal POVM exists, the answer can be obtained by considering two states which maximize 
${\gamma_{0_{lm}}'}$.
\end{proposition}
\begin{proof}
For simplicity, consider that ${\gamma_{0_{12}}'}$ is the maximum. Assume that there are states $i$ and $j$ with $\gamma_{0_{ij}}'<\gamma_{0_{12}}'$ that satisfy Helstrom conditions. It means that
\begin{eqnarray}
\gamma_{0_{ij}}'-p_l=|\boldsymbol{\gamma}'_{ij}-\boldsymbol{\tilde{v}_l}|, \enspace\enspace \text{for} \enspace\enspace l=i,j \\
\gamma_{0_{ij}}'-p_l > |\boldsymbol{\gamma}'_{ij}-\boldsymbol{\tilde{v}_l}|, \enspace\enspace \text{for} \enspace\enspace l\neq i,j 
\end{eqnarray}
then, writing the inequality for $l=1,2$ and summing them gives
\begin{eqnarray}
2\gamma_{0_{ij}}' &>&|\boldsymbol{\gamma}'_{ij}-\boldsymbol{\tilde{v}_1}|+|\boldsymbol{\gamma}'_{ij}-\boldsymbol{\tilde{v}_2}|+(p_1+p_2)\nonumber \\
&\geq& |\boldsymbol{\tilde{v}_1}-\boldsymbol{\tilde{v}_2}|+(p_1+p_2)=2\gamma_{0_{12}}' ,
\end{eqnarray}
the second inequality comes from the triangle inequality. So, we end up with $\gamma_{0_{ij}}' > \gamma_{0_{12}}'$, which is in contradiction with our first assumption $\gamma_{0_{ij}}' < \gamma_{0_{12}}'$.
\end{proof}

To proceed, we need to find two states that maximize Eq.~\eqref{eq:gamma0lm}. Then, we can test them using the following condition
\begin{equation} \label{eq:plgeqpi}
p_l+|\boldsymbol{\tilde{v}_l}-\boldsymbol{\gamma}|\geq p_i+|\boldsymbol{\tilde{v}_i}-\boldsymbol{\gamma}|,\enspace\enspace\enspace \forall{i}
\end{equation}
which is derived from the Helstrom condition, Eq.~\eqref{eq:gamma0-pgeq}. The cases $i=l$ and $m$ are trivial. If Eq.~\eqref{eq:gamma0lm} is maximized for more than one pair of states, we must test the Helstrom condition for each of them.
Note that for a typical ME discrimination problem, there may be multiple possible solutions, and we only need to find one of them to determine $\boldsymbol{\gamma}$. If the condition~\eqref{eq:plgeqpi} is not satisfied, we proceed to the next step.

{\it Step  (iii).---} In this step, we consider the three-state case. We label them as $l$, $m$, and $n$. From these three states, we construct $\boldsymbol{\gamma}'$ by drawing three hyperbolas corresponding to each pair of states $(lm, ln, mn)$. The resulting convex polytope is a triangle, and we only need to determine the point where two hyperbolas intersect ($\boldsymbol{\gamma}'$). This point must be inside the triangle. The corresponding probability can be calculated using Eq.~\eqref{eq:gamma0lmn}. As in the second step, we need to test the Helstrom condition for the three states that maximize this probability. The cases where $i=l$, $m$, or $n$ are trivial. Similar to the procedure in the previous step, we may have more than one three-state case to maximize Eq.~\eqref{eq:gamma0lmn}. In this case, we must test the Helstrom condition for each of them. However, if the Helstrom condition is not satisfied, we proceed to the next step.

{\it Step  (iv).---} This step is expected to yield the answer since we know that there is always a non-decomposable answer with a maximum of four states. In this case, although there are $\binom{4}{2}=6$ hyperbolas, it is enough that three hyperbolas, with a common focus point $\boldsymbol{\tilde{v}_m}$, to intersect at a single point inside the polytope to obtain $\boldsymbol{\gamma}'$. The related probability can be obtained from Eq.~\eqref{eq:gamma0lmnk}. The four-state case that maximizes this probability must be the states we are searching for to find $\boldsymbol{\gamma}$. Note that, similar to the previous steps, this four-state case that maximizes Eq.~\eqref{eq:gamma0lmnk} may not be unique.\\
The deriving of Eq.~\eqref{eq:gamma0lmn} and Eq.~\eqref{eq:gamma0lmnk} for steps (iii) and (iv) respectively will be discussed in the next section.\\

As discussed earlier, we do not need to consider more than four qubits/steps to find $\Gamma$. An alternative explanation is stated as follows:\\ considering the Bloch representation of $\Gamma$, Eq.~\eqref{eq:defGamma}, it is evident that there are only four unknowns to be determined: a scalar $\gamma_0$ and a three-dimensional vector $\boldsymbol{\gamma}$. These can be easily determined by considering four qubit states along with their corresponding Bloch vectors (See Eq.~\eqref{eq:gamma0-p}.\\
Although we introduced a way to find the necessary states that lead to the Lagrange operator, the fourth step seems to be less likely. This is because it is often expected to have at least one non-decomposable POVM answer with fewer elements for a given probability.

At this stage, it is worth highlighting some notable results:
\begin{itemize}

\item According to the above instruction, Eq.~\eqref{eq:pl-pm} plays a crucial role in finding $\gamma$. This equation remains unchanged as long as both the relative distances between the states $\tilde{d}_{lm}$'s and the prior probabilities $p_i$'s are consistent. It implies that by translation or rotation of the polytope, or equivalently all $\boldsymbol{\tilde{v}_i}$'s as a whole, with a fixed vector inside the Bloch sphere, will result in a consistent guessing probability. This class is characterized by an \emph{unchanged guessing probability.}
\item Furthermore, if we perform a translation and rescaling of the polytope inside the Bloch sphere while keeping the $p_i$ values unchanged, the measurement operators defined by Eq.~\eqref{eq:defpi} will also remain unchanged. This class is also characterized by an \emph{unchanged measurement operators.}
\item Based on the aforementioned instruction, to obtain a complete solution for a general qubit state discrimination problem, we need to find all non-decomposable sets, $E'=\{\boldsymbol{\hat{n}_1},\boldsymbol{\hat{n}_2},\cdots \}$, of the problem. After finding $\boldsymbol{\gamma}$, it can be simply accomplished by looking at the geometry of the $\boldsymbol{\hat{n}_i}$'s (see the explanation after theorem~(\ref{theorem:decomposable})). Then, we can obtain all of the possible optimal POVMs using convex combinations of all non-decomposable sets, i.e., $\sum_{i=1}^k \beta_i M'_i$, where $k$ is the number of all non-decomposable sets, $0\leq \beta_i \leq 1$, and $\sum_{i=1}^k \beta_i=1$ to satisfy completeness relation of a POVM measurement.\\
For this purpose, in the next section, we will analytically solve the problem for these cases. Some of these results, such as the problem of two-state discrimination, have been known for a long time.
\end{itemize}
\section{Solution in qubit state discrimination for $N\leq4$}\label{section:Solution in qubit state discrimination}
\subsection{Two-state case}\label{subsection:Two States Case}
As the first example, let us consider the case of two arbitrary qubit states $\rho_1$ and $\rho_2$ with priori probabilities $p_1\geq p_2$. We will follow the steps described above. First, if $\Gamma=\tilde{\rho}_1$, then the measurement operators are given by $\pi_1=\mathbb{1}$ and $\pi_2=0$, and the guessing probability is $P_{\textrm{guess}}=\gamma_0=p_1$. However, if this condition is not met, we know that from Eq.~\eqref{eq:sigmaalpha} $\alpha_1=\alpha_2=1$, and the corresponding POVMs are projectors given by Eq.~\eqref{eq:defpi} with
$\boldsymbol{\hat{n}_1}=-\boldsymbol{\hat{n}_2}=\frac{\tilde{v}_1-\tilde{v}_2}{\tilde{d}_{12}}$ and $\tilde{d}_{12}=|\boldsymbol{\tilde{v}_1}-\boldsymbol{\tilde{v}_2}|$.
The associated guessing probability is then given by Eq.~\eqref{eq:gamma0lm}
\begin{equation}
P_{\textrm{guess}}=\frac{1}{2}(1+\tilde{d}_{12}),
\end{equation}
which is the well-known Helstrom relation for two states \cite{HelstromBook1976}.
\subsection{Three-state case} \label{subsection:Three States}
We now consider a general case of three arbitrary qubit states with prior probabilities $p_1\geq p_2\geq p_3$. To proceed, note that the discrimination of an arbitrary set of three-qubit states can be reduced to discrimination of three-qubit states,  all embedded in $x-z$ plane, defined by
\begin{eqnarray}
\rho_1 &=&\frac{1}{2}
\begin{pmatrix}
1+a & 0 \\
0 & 1-a \\
\end{pmatrix} , \nonumber
\\
\rho_2 &=&\frac{1}{2}
\begin{pmatrix}
1+b\cos{\theta} & b\sin{\theta} \\
b\sin{\theta} & 1-b\cos{\theta} \\
\end{pmatrix} , \nonumber
\\
\rho_3 &=&\frac{1}{2}
\begin{pmatrix}
1+c\cos{\phi} & -c\sin{\phi} \\
-c\sin{\phi} & 1-c\cos{\phi} \\
\end{pmatrix} .
\label{three states density matrices}
\end{eqnarray}
where $\theta$ and $\phi$ are illustrated in Fig.~\ref{fig:three states}. The proof is provided in Appendix~\ref{section:The Rotation Matrix}. Then, the corresponding Bloch vectors $\boldsymbol{\tilde{v}_i}$'s are given by
\begin{eqnarray}
\boldsymbol{\tilde{v}_1}&=&(0,ap_1) ,\nonumber \\
\boldsymbol{\tilde{v}_2}&=& (bp_2\sin{\theta},bp_2\cos{\theta}) , \nonumber \\
\boldsymbol{\tilde{v}_3}&=&(-cp_3\sin{\phi},cp_3\cos{\phi}),
\end{eqnarray}
where $(x,z)$ is a simplified representation for $(x,0,z)$.

\begin{figure}[h]
\mbox{\includegraphics[scale=0.55]{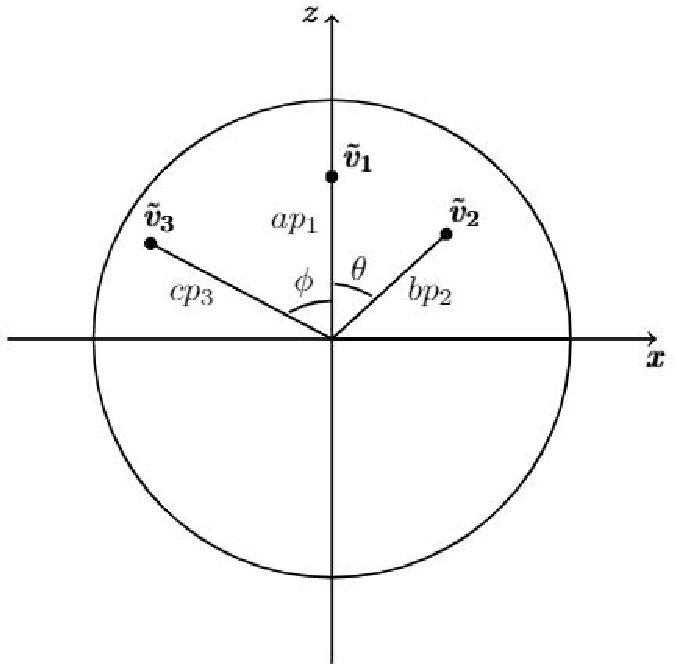}}
\caption{The Bloch representation of the three states $\rho_1$, $\rho_2$, and $\rho_3$ on the $x-z$ plane, with their corresponding priori probabilities $p_1\geq p_2\geq p_3$. The circle is the intersection of the Bloch sphere with the $x-z$ plane.}
\label{fig:three states}
\end{figure}

With the assumption that $p_1$ is the greatest priori probability, first, we have to check whether $\tilde{\rho}_1$ is equal to  $\Gamma$ or not. To do this, we refer to Eq.~\eqref{pj-pi} and check the following conditions
\begin{eqnarray}\label{nomeasurement1}
\bigl|(-bp_2\sin{\theta},ap_1-bp_2\cos{\theta})\bigr|&\leq& p_{12},  \\ \label{nomeasurement2}
\bigl|(cp_3\sin{\phi},ap_1-cp_3\cos{\phi})\bigr|&\leq& p_{13}.
\end{eqnarray}
Here, $|(x,z)^{\textrm{t}}|=\sqrt{x^2+z^2}$. If the above conditions are not satisfied, we must calculate $P_{\textrm{guess}}$ from Eq.~\eqref{eq:gamma0lm} for each pair of states, and then check the Helstrom condition using Eq.~\eqref{eq:plgeqpi}. The explicit form of these conditions to detect one of the pairs $(\rho_1,\rho_2)$, $(\rho_1,\rho_3)$ or $(\rho_2,\rho_3)$ ,respectively, are
\begin{eqnarray}\label{two-elements12}
\bigl|(&x_{12}'&+cp_3\sin{\phi}, z_{12}'-cp_3\cos{\phi})\bigr|\leq \tilde{R}_{12}+p_{13}, \quad
\\ \label{two-elements13}
\bigl|(&x_{13}'&-bp_2\sin{\theta}, z_{13}'-bp_2\cos{\theta})\bigr|\leq \tilde{R}_{13}+p_{12},\quad
\\ \label{two-elements23}
\bigl|(&x_{23}'&, z_{23}'-ap_1)\bigr|\leq \tilde{R}_{23}-p_{12}.
\end{eqnarray}
with the corresponding $\boldsymbol{\gamma}_{ij}'=(x_{ij}',z_{ij}')^{\textrm{t}}$ for these states obtained from Eq.~\eqref{eq:gammalm}, we have
\begin{eqnarray} \label{eq:12-23-13}
x_{12}'&=&\frac{bp_2\tilde{R}_{12}\sin{\theta}}{2 \tilde{R}_{12}+p_{12}} ,
\\
z_{12}'&=&a p_1-\frac{(a p_1 -bp_2\cos{\theta})\tilde{R}_{12}}{2 \tilde{R}_{12}+p_{12}} , \nonumber \\
x_{23}'&=&bp_2\sin{\theta}-\frac{(bp_2\sin{\theta}+cp_3\sin{\phi})\tilde{R}_{23}}{2\tilde{R}_{23}+p_{23}} , \nonumber \\
z_{23}'&=&bp_2\cos{\theta}  -\frac{(bp_2\cos{\theta}-cp_3\cos{\phi})\tilde{R}_{23}}{2 \tilde{R}_{12}+p_{23}} , \nonumber \\
\nonumber
x_{13}'&=&-\frac{cp_3\tilde{R}_{13}\sin{\phi}}{2\tilde{R}_{13}+(p_{13})} ,
\\
z_{13}'&=&ap_1-\frac{(ap_1-cp_3\cos{\phi})\tilde{R}_{13}}{2 \tilde{R}_{13}+(p_{13})} , \nonumber
\end{eqnarray}
where $p_{ij}:=p_i-p_j$ and $\tilde{R}_{lm}$ is defined in Eq.~\eqref{eq:Rlm}.\\
For any specific problem with known values of $p_i$, $\theta$, and $\phi$, it is important to note that each condition in Eqs.~\eqref{two-elements12} to \eqref{two-elements23} that is satisfied will correspond to the unique solution of the problem.
Accordingly, the optimal POVM elements are given by $\pi_i=\frac{1}{2}(\mathbb{1}+\boldsymbol{\hat{n}_i} \cdot \boldsymbol{\sigma})$, as defined in Eq.~\eqref{eq:defpi}. The unit vectors corresponding to the non-null measurement elements can be written as
\begin{equation} \label{nlm}
\boldsymbol{\hat{n}_l}=-\boldsymbol{\hat{n}_m}=\frac{\boldsymbol{\tilde{v}_l}-\boldsymbol{\tilde{v}_m}}{\tilde{d}_{lm}}
\end{equation}

Finally, if none of the aforementioned conditions are met, the solution for $\boldsymbol{\gamma}$ must be found while all three states are detectable. In other words, three hyperbolas should intersect at a single point. Note that since three points are embedded in a plane, it is sufficient to find the intersection between two hyperbolas in this plane. So, by using the properties of hyperbolas, the guessing probability can be obtained as 

\begin{figure}[h]
\mbox{\includegraphics[scale=0.42]{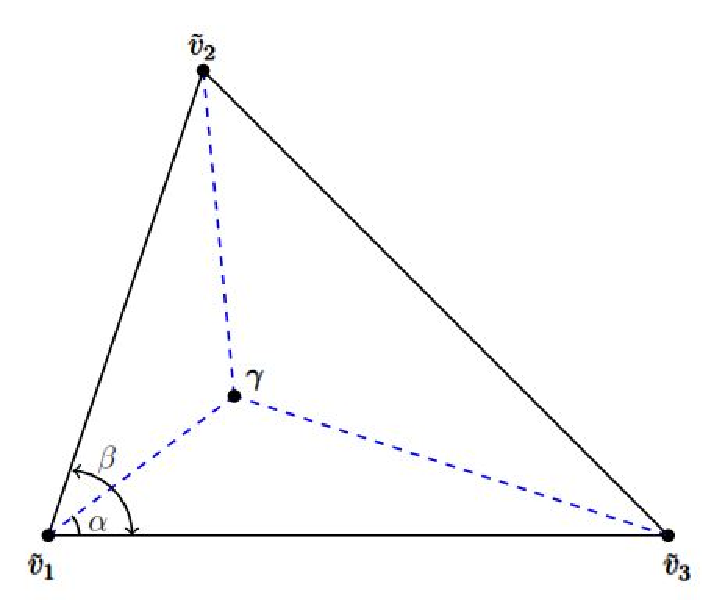}}
\caption{The Bloch representation depicts the three-state case, where all three states correspond to a three-element non-decomposable POVM.}
\label{fig:triangle}
\end{figure}

\begin{eqnarray}
P_{\textrm{guess}}=\gamma_{0}&=&p_1+|{\boldsymbol{\gamma}}-\boldsymbol{\tilde{v}_1}| \nonumber\\
&=&p_1+\frac{\tilde{d}_{12}^2-p_{12}^2}{2(\tilde{d}_{12}\cos(\beta-\alpha)+p_{12})},\label{eq:gamma0lmn}
\end{eqnarray}
where 

\begin{eqnarray}
\alpha= \arccos(\frac{h_0h_2+h_1\sqrt{h_0^2+h_1^2-h_2^2}}{h_0^2+h_1^2})
\end{eqnarray}

and
\begin{equation}
\beta=\arccos(\frac{\tilde{d}_{12}^2+\tilde{d}_{13}^2-\tilde{d}_{23}^2}{2\tilde{d}_{12}\tilde{d}_{13}}).
\end{equation}
where
\begin{eqnarray}
h_0&=&\tilde{d}_{12}(\tilde{d}_{13}^2-p_{13}^2)\cos(\beta)-\tilde{d}_{13}(\tilde{d}_{12}^2-p_{12}^2), \nonumber \\
h_1&=&\tilde{d}_{12}(\tilde{d}_{13}^2-p_{13}^2)\sin(\beta), \nonumber \\ 
h_2&=&p_{13}(\tilde{d}_{12}^2-p_{12}^2)-p_{12}(\tilde{d}_{13}^2-p_{13}^2).
\end{eqnarray}
To find the vector $\boldsymbol{\gamma}$, we can use the geometric of a triangle as illustrated in Fig.~\ref{fig:triangle}. By employing the Gram-Schmidt method, we can use two edges of the triangle, namely $\boldsymbol{\tilde{v}_2}-\boldsymbol{\tilde{v}_1}$ and $\boldsymbol{\tilde{v}_3}-\boldsymbol{\tilde{v}_1}$,  to construct an orthonormal basis within the plane where the triangle lies. Let us denote this orthonormal basis as $\boldsymbol{\hat{k}_1}$ and $\boldsymbol{\hat{k}_2}$. Then, we have
\begin{equation}
\boldsymbol{\hat{k}_1}=\frac{\boldsymbol{\tilde{v}_2}-\boldsymbol{\tilde{v}_1}}{\tilde{d}_{12}}\quad , \qquad \boldsymbol{\hat{k}_2}=\frac{\vec{f}}{|f|} 
\end{equation}
where
\begin{equation}
\vec{f}=(\boldsymbol{\tilde{v}_3}-\boldsymbol{\tilde{v}_1}) - \boldsymbol{\hat{k}_1} [\frac{(\boldsymbol{\tilde{v}_2}-\boldsymbol{\tilde{v}_1})}{\tilde{d}_{12}}.(\boldsymbol{\tilde{v}_3}-\boldsymbol{\tilde{v}_1})]. \nonumber
\end{equation}
Hence, $\boldsymbol{\gamma}$ can be expressed as follows
\begin{equation}
\boldsymbol{\gamma}=\frac{\tilde{d}_{12}^2-p_{12}^2}{2(\tilde{d}_{12}\cos(\beta-\alpha)+p_{12})} (\cos(\beta-\alpha)\boldsymbol{\hat{k}_1}+\sin(\beta-\alpha)\boldsymbol{\hat{k}_2})+\boldsymbol{\tilde{v}_1},
\end{equation}
So, POVM elements can be obtained using Eqs.~\eqref{eq:defpi} and \eqref{eq:ni}. 

It is important to emphasize that before obtaining the optimal operators, we perform transformations on the Bloch vectors using a rotation matrix (detailed in Appendix~\ref{section:The Rotation Matrix}) and a displacement.\\In case of $\rho_i$ and $\rho_j$ to be two detectable states, optimal operators are $\pi_i=\rho_{\boldsymbol{\hat{n}_i}}$ and $\pi_j=\rho_{-\boldsymbol{\hat{n}_i}}$, where $\boldsymbol{\hat{n}_i}=\frac{\boldsymbol{\tilde{v}_i}-\boldsymbol{\tilde{v}_j}}{\tilde{d}_{ij}}$. So, the answers for the optimal operators $\{\pi_i\}$ can be easily obtained using the vectors of the original problem. However, for the case in which all three states are detectable, we need to obtain $\boldsymbol{\hat{n}_i}$'s and then rotate them back to the original problem using the inverse of the rotation matrix. Afterward, we can use the formula $\pi_i=\rho_{\boldsymbol{\hat{n}_i}}$ to obtain the optimal operators for the original problem. \\
It is worth mentioning that the result we obtained here is in perfect agreement with the findings of \cite{HaPRA2013}, which were obtained through a different approach that involved solving the duality problem. This agreement provides strong support for the validity and reliability of our results, further highlighting the consistent relationship between the dual variables and the optimization objective.
\subsubsection*{Trine states}
Here, for more illustration of this case, let us consider the discrimination among trine states, which refers to the qubit states associated with equidistant points on the surface of the Bloch sphere. These states are defined by
\begin{eqnarray}
\ket{\psi_1}&=&\frac{1}{\sqrt{2}} (\ket{0}+\ket{1}) , \nonumber \\
\ket{\psi_2}&=&\frac{1}{\sqrt{2}} (\ket{0}+e^{+\frac{2\pi}{3} i}\ket{1}) , \nonumber \\
\ket{\psi_3}&=&\frac{1}{\sqrt{2}} (\ket{0}+e^{-\frac{2\pi}{3} i}\ket{1}),
\end{eqnarray}
Then, based on the previous discussion, it is possible to rotate these states on the Bloch sphere to align them on the $x-z$ plane
\begin{eqnarray}
\ket{\psi'_1}&=&\ket{0} , \nonumber \\
\ket{\psi'_2}&=&\frac{1}{2} (\ket{0}+\sqrt{3}\ket{1}) , \nonumber \\
\ket{\psi'_3}&=&\frac{1}{2} (\ket{0}-\sqrt{3}\ket{1}).
\end{eqnarray}
By writing the corresponding density matrices and comparing them with Eq.~\eqref{three states density matrices}, we find that $a=b=c=1$ and $\theta=\phi=2\pi/3$.

To proceed, let us consider the case with equal prior probabilities. Since the three states are pure and located on the surface of the Bloch sphere, finding the answer is straightforward. According to the discussion in Sect.~\ref{subsection:Discrimination of qubit states} (or more straightforward, later, from Appendix.~\ref{Qubit states: Equal priori probabilities}), the guessing probability is $P_{\textrm{guess}}=\frac{2}{3}$ and the corresponding POVM operators for the original problem are $\pi_i=\frac{2}{3}\ket{\psi_i}\bra{\psi_i}$, known as the trine measurement \cite{BanIJTP1997}.

Next, we consider the case with arbitrary prior probabilities. We use the same parameterization of Weir \etal \cite{WeirQST2018} for the prior probabilities: $p_1=p+\delta$, $p_2=p-\delta$, and $p_3=1-2p$. Here, it is assumed that $0<p_3\le p_2\le p_1$, which implies that $\frac{1}{3}\leq p \leq\frac{1}{2}$ and $0\le \delta\le \min\{3p-1, p\}$.

Using the method described in  Sect.~\ref{subsection:Qubit states: Arbitrary priori probabilities}, we can conclude that the ``no measurement strategy'' cannot be optimal for $\gamma_0=p_1$ and $\boldsymbol{\gamma}=p_1\boldsymbol{\hat{z}}$. It is obvious from our analysis that the condition $p_{1i}\geq \tilde{d}_{1i}$ is not satisfied for all $i$. The same result is obtained from Eqs.~\eqref{nomeasurement1} and \eqref{nomeasurement2} as well.
So, we must have either two or three POVM measurement elements. For the two-state case, it can be observed that Eq.~\eqref{eq:gamma0lm} is maximum for states $\rho_1$ and $\rho_2$.
In this case, Eq.~\eqref{two-elements12} yields
\begin{equation} \label{eq:solvedelta}
\frac{1}{2}\sqrt{(\frac{2\delta^2}{\tilde{d}_{12}}+2-3p)^2+3(\delta+\frac{2p\delta}{\tilde{d}_{12}})^2}\le \frac{\tilde{d}_{12}}{2}-1+3p
\end{equation}
where $\tilde{d}_{12}=\sqrt{3p^2+\delta^2}$. By solving Eq.~\eqref{eq:solvedelta} for $\delta$, we obtain four roots, of which the only valid values for $\frac{1}{3}\le p\le \frac{1}{2}$ and $ 0\le \delta$ are given by
\begin{equation} \label{eq:delta}
\delta \leq \left(2-6p+5p^2-2(1-2p)\sqrt{4p^2-2p+1}\right)^{\frac{1}{2}}.
\end{equation}
In this region of the two-state case, the guessing probability is
\begin{equation} \label{eq:pguessdelta}
P_{\textrm{guess}}=\frac{1}{2}\sqrt{3p^2+\delta^2}+p,
\end{equation}

In the region where the POVM measurement needs to be a three-element POVM, we can use Eq.~\eqref{eq:pl-pm} for three states to find $\boldsymbol{\gamma}=(\gamma_x,\gamma_y,\gamma_z)$. By employing these equations and performing algebraic calculations, we can obtain the desired $\boldsymbol{\gamma}$ as follows
\begin{eqnarray}
\gamma_x&=&\frac{2\sqrt{3}(1-2p)(p-\delta)(p+\delta)^2(\delta-3p+1)}{9p^4-4p^3+6p^2\delta^2-12p\delta^2+4\delta^2+\delta^4},  \\
\gamma_y&=&0, \nonumber \\
\gamma_z&=&\frac{2(1-2p)(p^2-\delta^2)\bigl(-3p^2+(6\delta+1)p+\delta^2-3\delta \bigr)}{9p^4-4p^3+6p^2\delta^2-12p\delta^2+4\delta^2+\delta^4}. \nonumber
\end{eqnarray}
Using $P_{\textrm{guess}}=\gamma_0=p_i+|\boldsymbol{\tilde{v}_i}-\boldsymbol{\gamma}|$ or Eq.~\eqref{eq:gamma0lmn}, the guessing probability can be calculated
\begin{eqnarray}
P_{\textrm{guess}}&=&\gamma_0=p+\delta+\sqrt{\gamma_{x}^2+(\gamma_z-p-\delta)^2} \nonumber\\
&=&\frac{2(1-2p)(p^2-\delta^2)(3p^2+\delta^2-2p)}{9p^4-4p^3+6p^2\delta^2-12p\delta^2+4\delta^2+\delta^4}, \nonumber
\end{eqnarray}
These equations, namely \eqref{eq:delta}, \eqref{eq:pguessdelta}, and the derived expression for $P_{\textrm{guess}}$, are in complete agreement with \cite{WeirQST2018} which was obtained using a different approach.\\

\subsection{Four-state case}\label{subsection:Four States}
In this section, we consider a general problem of four qubit states. According to the previous discussion, we can divide this problem into two cases:\\ (i) When the four states form a two-dimensional convex polytope, according to Caratheodory's theorem, the vector $\boldsymbol{\gamma}$ can be written as a convex combination of at most three points. Therefore, the optimal measurement has at most three non-null elements. In this case, we can solve the problem by following the three-step instruction provided in Sect.~\ref{subsection:Qubit states: Arbitrary priori probabilities}. \\ (ii)When the four states form a three-dimensional polytope, specifically a tetrahedron, a four-element optimal POVM may exist. To find the guessing probability and optimal measurement in this case, we can use the properties of tetrahedrons. A tetrahedron consists of four triangular faces, six edges, and four vertices. Depending on the location of $\boldsymbol{\gamma}$, we can specify the number of detectable states optimally. If $\boldsymbol{\gamma}$ is on one vertex of the tetrahedron, the optimal answer is the no measurement strategy, and the other three states are unguessable. If $\boldsymbol{\gamma}$ lies on one edge, the optimal measurement will be a POVM with two non-null elements, allowing for two guessable states. Furthermore, if $\boldsymbol{\gamma}$ lies on one of the faces, three states can be detected through an optimal measurement, while one state will remain unguessable. Finally, if $\boldsymbol{\gamma}$ is an interior point of the tetrahedron, all four states are guessable, meaning that the six hyperbolas meet at a single point. To find the guessing probability for this case, let us consider the vector $\boldsymbol{\gamma}$ as an interior point in the tetrahedron $\boldsymbol{\tilde{v}_1} \boldsymbol{\tilde{v}_2} \boldsymbol{\tilde{v}_3} \boldsymbol{\tilde{v}_4}$, therefore the guessing probability can be written as 
%\small
\begin{widetext}
\begin{eqnarray}\label{eq:gamma0lmnk}
P_{\textrm{guess}}=p_1+ 
\frac{-2D_{123}p_{12}(\tilde{d}_{12}^2-p_{12}^2)\cos^2{\theta}-l_0l_1-\cos{\theta}\sqrt{F}}{4 D_{123}(\tilde{d}_{12}^2-p_{12}^2)\cos^2{\theta}+l_1^2},
\label{pguess4}
\end{eqnarray}
where the quantities $l_0$, $l_1$, $D_{123}$ and $F$ are given by
\begin{eqnarray} 
l_0&=&\tilde{d}_{12}^2(\tilde{d}_{13}^2+\tilde{d}_{23}^2-\tilde{d}_{12}^2+p_{12}^2-2p_{13}^2)+p_{12}^2(\tilde{d}_{13}^2-\tilde{d}_{23}^2), \\
l_1&=&2 p_{12}(\tilde{d}_{12}^2+\tilde{d}_{13}^2-\tilde{d}_{23}^2)-4 p_{13} \tilde{d}_{12}^2, \nonumber \\
D_{123}&=&-(\tilde{d}_{12}+\tilde{d}_{13}+\tilde{d}_{23})(\tilde{d}_{12}+\tilde{d}_{13}-\tilde{d}_{23})(\tilde{d}_{23}+\tilde{d}_{12}-\tilde{d}_{13})(\tilde{d}_{13}+\tilde{d}_{23}-\tilde{d}_{12}), \nonumber \\
F&=&\bigl(4\cos^2{\theta}D_{123}(\tilde{d}_{12}^2-p_{12}^2)\tilde{d}_{12}^2+l_1^2(\tilde{d}_{12}^2-p_{12}^2)+4p_{12}l_0 l_1-4 l_0^2\bigr)D_{123}(\tilde{d}_{12}^2-p_{12}^2),
\label{l0l1}
\end{eqnarray}
\end{widetext}
%\normalsize
and $\theta$ is the dihedral angle of the edge $\boldsymbol{\tilde{v}_1} \boldsymbol{\tilde{v}_2}$ in the tetrahedron $\tilde{\boldsymbol{\gamma}} \boldsymbol{\tilde{v}_1}\boldsymbol{\tilde{v}_2} \boldsymbol{\tilde{v}_3}$ (Fig.~\ref{fig:tetrahedron}). In appendix \ref{section:Tetrahedron geometry}, we demonstrate how to determine this angle, along with the mathematical framework we used to derive Eq.~\eqref{pguess4} for the problem of four qubit states.
\begin{figure}[t]
\mbox{\includegraphics[scale=0.55]{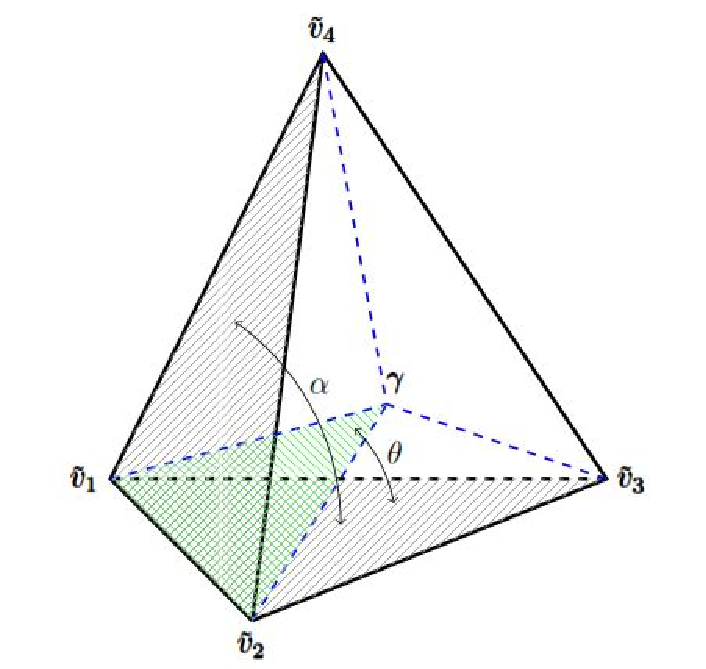}}
\caption{Representation of the tetrahedron formed by a four-state case, with $\theta$ and $\alpha$ as the dihedral angles.}
\label{fig:tetrahedron}
\end{figure}

To find the optimal POVM in this case, we need to first find the vector $\boldsymbol{\gamma}$ and then utilize  Eqs.~\eqref{eq:defpi} and \eqref{eq:ni} to obtain its elements. To find the vector $\boldsymbol{\gamma}$, we begin with constructing an orthonormal basis using the vectors  $\boldsymbol{\tilde{v}_i}$ 's. We consider three linearly independent vectors $\boldsymbol{\tilde{v}_1}-\boldsymbol{\tilde{v}_2}$, $\boldsymbol{\tilde{v}_1}-\boldsymbol{\tilde{v}_3}$ and $\boldsymbol{\tilde{v}_1}-\boldsymbol{\tilde{v}_4}$, and apply the Gram-Schmidt process as follows

\begin{eqnarray}
\boldsymbol{\hat{k}_1}=\frac{\vec{f_1}}{|\vec{f_1}|} , \quad \boldsymbol{\hat{k}_2}=\frac{\vec{f_2}}{|\vec{f_2}|}, \quad  \boldsymbol{\hat{k}_3}=\frac{\vec{f_3}}{|\vec{f_3}|},
\end{eqnarray}
where
\begin{eqnarray}
   \vec{f_1}&=& \frac{\boldsymbol{\tilde{v}_2}-\boldsymbol{\tilde{v}_1}}{\tilde{d}_{12}}, \\
   \vec{f_2}&=& (\boldsymbol{\tilde{v}_1}-\boldsymbol{\tilde{v}_3})-[(\boldsymbol{\tilde{v}_1}-\boldsymbol{\tilde{v}_3}).\boldsymbol{\hat{k}_1}]\boldsymbol{\hat{k}_1}, \nonumber \\
   \vec{f_3}&=& (\boldsymbol{\tilde{v}_1}-\boldsymbol{\tilde{v}_4})-[(\boldsymbol{\tilde{v}_1}-\boldsymbol{\tilde{v}_4}).\boldsymbol{\hat{k}_1}]\boldsymbol{\hat{k}_1}-[(\boldsymbol{\tilde{v}_1}-\boldsymbol{\tilde{v}_4}).\boldsymbol{\hat{k}_2}]\boldsymbol{\hat{k}_2}. \nonumber
\end{eqnarray}
Here, the set $\{\boldsymbol{\hat{k}_1} , \boldsymbol{\hat{k}_2},\boldsymbol{\hat{k}_3}\}$ forms an orthonormal basis. Then, we can express $\boldsymbol{\gamma}$ as
\begin{equation}
\boldsymbol{\gamma}=(\boldsymbol{\gamma}.\boldsymbol{\hat{k}_1})\boldsymbol{\hat{k}_1}+(\boldsymbol{\gamma}.\boldsymbol{\hat{k}_2})\boldsymbol{\hat{k}_2}+(\boldsymbol{\gamma}.\boldsymbol{\hat{k}_3})\boldsymbol{\hat{k}_3},
\end{equation}
To find the coefficients $\boldsymbol{\gamma}.\hat{k_i}$, we can start with $\boldsymbol{\gamma}.\boldsymbol{\hat{k}_1}$
\begin{eqnarray}
\boldsymbol{\gamma}.\boldsymbol{\hat{k}_1}&=&\boldsymbol{\gamma}.\frac{\vec{f_1}}{|\vec{f_1}|} \\
&=& \frac{1}{|\vec{f_1}|} \boldsymbol{\gamma}.(\boldsymbol{\tilde{v}_2}-\boldsymbol{\tilde{v}_1})\nonumber \\
&=&  \frac{1}{|\vec{f_1}|} (\boldsymbol{\gamma}-\boldsymbol{\tilde{v}_1}).(\boldsymbol{\tilde{v}_2}-\boldsymbol{\tilde{v}_1}) +\boldsymbol{\tilde{v}_1}.(\boldsymbol{\tilde{v}_2}-\boldsymbol{\tilde{v}_1}), \nonumber
\end{eqnarray}
By using $|\boldsymbol{\gamma}-\boldsymbol{\tilde{v}_1}|$ from Eq.~\eqref{pguess4}, we obtain
\begin{equation}
(\boldsymbol{\gamma}-\boldsymbol{\tilde{v}_1}).(\boldsymbol{\tilde{v}_2}-\boldsymbol{\tilde{v}_1})=|\boldsymbol{\gamma}-\boldsymbol{\tilde{v}_1}|\tilde{d}_{12}\cos(\delta_{\gamma\tilde{v}_1\tilde{v}_2}),
\end{equation}
where $\delta_{\gamma\tilde{v}_1\tilde{v}_2}$ is the angle between $(\boldsymbol{\gamma}-\boldsymbol{\tilde{v}_1})$ and $(\boldsymbol{\tilde{v}_1}-\boldsymbol{\tilde{v}_2})$ in triangle $\boldsymbol{\gamma}\boldsymbol{\tilde{v}_1}\boldsymbol{\tilde{v}_2}$ that can be easily obtained using the rules of sines in triangle. Then, we will have $\boldsymbol{\gamma}.\boldsymbol{\hat{k}_1}$. The same procedure can be applied to find $\boldsymbol{\gamma}.\boldsymbol{\hat{k}_2}$ and $\boldsymbol{\gamma}.\boldsymbol{\hat{k}_3}$.

\begin{example} 
As an example, let us consider the problem of four symmetric qubit states with the same purity $\zeta$ ($0<\zeta\leq 1$) and unequal a priori probabilities that are parametrized by parameters $\alpha$ and $\beta$ as $p_1=\frac{1}{4}+\alpha+\beta$, $p_2=\frac{1}{4}+\alpha-\beta$, $p_3=\frac{1}{4}-\alpha+\beta$, and $p_4=\frac{1}{4}-\alpha-\beta$ which form a polytope having a volume in the Bloch sphere, i.e. a regular tetrahedron. In this case, the states are expressed symmetrically as four points on the Bloch sphere.
The related Bloch vectors are $\boldsymbol{{v}_1}=\frac{\zeta}{\sqrt{3}}(1, 1, 1)$, $
\boldsymbol{{v}_2}=\frac{\zeta}{\sqrt{3}}(1, -1, -1)$, $\boldsymbol{{v}_3}=\frac{\zeta}{\sqrt{3}}(-1, 1, -1)$, and $\boldsymbol{{v}_4}=\frac{\zeta}{\sqrt{3}}(-1, -1, 1)$. Then $\boldsymbol{\tilde{v}_i}$'s are 
\begin{eqnarray}
\boldsymbol{\tilde{v}_1}&=&\frac{p_1\zeta}{\sqrt{3}}(1, 1, 1), \nonumber\\
\boldsymbol{\tilde{v}_2}&=&\frac{p_2\zeta}{\sqrt{3}}(1, -1, -1),\nonumber \\
\boldsymbol{\tilde{v}_3}&=&\frac{p_3\zeta}{\sqrt{3}}(-1, 1, -1), \nonumber\\
\boldsymbol{\tilde{v}_4}&=&\frac{p_4\zeta}{\sqrt{3}}(-1, -1, 1), 
\label{foursymmetricpurestates}
\end{eqnarray}

If $\alpha=\beta=0$ the states are equidistant in this example with $\tilde{d}_{ij}=\frac{\zeta}{\sqrt{6}}$ for $i\neq j$, $l_0=\tilde{d}_{12}^4$, $l_1=0$ , and $D_{123}=-3\tilde{d}_{12}^4$. After some simple calculations, using the method presented in Appendix \ref{section:Tetrahedron geometry}, the angle $\theta_{12}$ can be obtained as $\theta=\arccos(\sqrt{\frac{2}{3}})$. Putting these values in Eq. \eqref{pguess4}, we end up with $P_{\textrm{guess}}=\frac{\zeta}{2}$. This result is consistent with the result of Corollary~\ref{corollary:7} for states with equal a priori probability, as well as with the result obtained in \cite{BaeNJP2013}.
\begin{figure}[t]
\mbox{\includegraphics[scale=0.45]{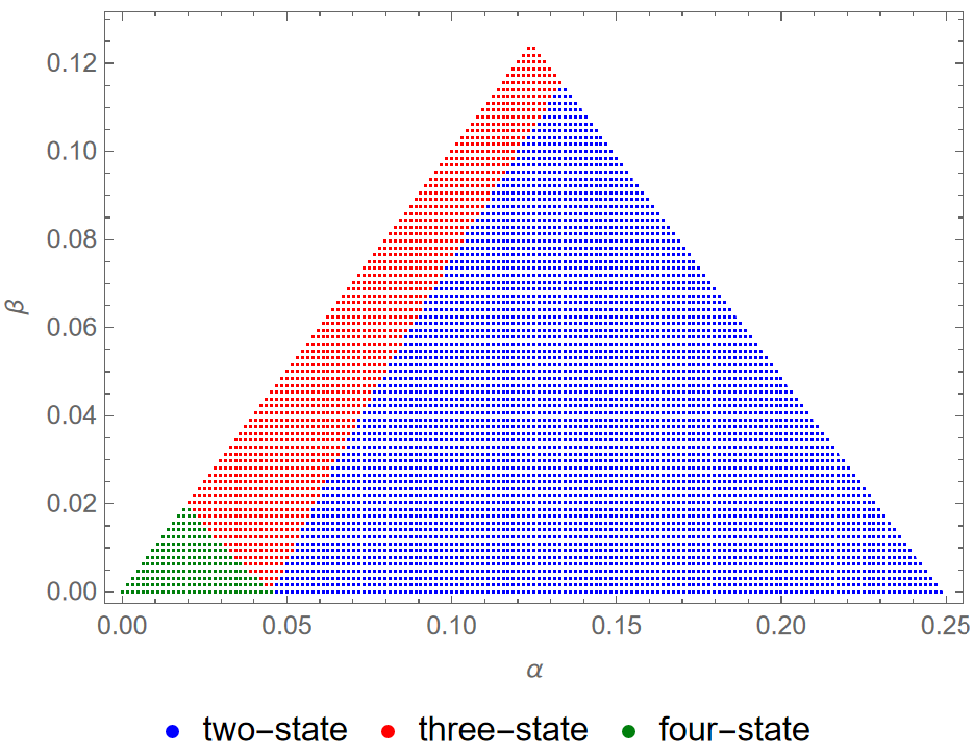}}
\caption{The simplest non-decomposable answers of four qubit states, which have the same purity $\zeta$, and different prior probabilities, which depend on the defined parameters $\alpha$ and $\beta$ in the example. When $\alpha$ and $\beta$ tend to be zero, the answer is obtained by considering all four states (green area). As $\alpha$ and $\beta$ increase, $\boldsymbol{\gamma}$ can be determined either by the two-state or three-state case, which corresponds to the blue and red areas, respectively.}
\label{fig:4states}
\end{figure}
For other cases, depending on the values of $\alpha$ and $\beta$, we can have different optimal answers, including two-element and three-element measurements. In Fig.~\ref{fig:4states}, different values of $\alpha$ and $\beta$ are shown, corresponding to the simplest non-decomposable answer. In this specific case, since four vectors in Eq. \eqref{foursymmetricpurestates} form a tetrahedron, all answers are unique.
\end{example}
Knowing the optimal answers for $N\leq4$, we have all the necessary information to solve the general problem of $N$ qubit states. For $N>4$ cases, since the optimal answer will be fulfilled with a non-decomposable optimal POVM with at most four non-null elements, we can use the guessing probabilities that were obtained for the cases of $N\leq4$. For this purpose, the guessing probability of a problem with $N$ states,$\{p_i,\rho_i\}_{i=1}^N$, can be rewritten as
\begin{equation}
P_{\textrm{guess}}^{N}=(\sum_{i=1}^{N_s}p_{i}) P_{\textrm{guess}}^S (\frac{p_i}{\sum_{i=1}^{N_s}p_{i}},\rho_i), \enspace\enspace\enspace \rho_i\in S
\end{equation} 
where $S$ is a subset of $\{\rho_i\}_{i=1}^N$ with $N_S$ elements ($N_s\leq4$). 
$P_{\textrm{guess}}^S$ in this relation can be obtained by using Eqs. \eqref{eq:gamma0lm}, \eqref{eq:gamma0lmn}, and \eqref{pguess4}.

\section{$N>4$ cases} {\label{Ncases}}
We begin this section with the following theorem:

\begin{theorem}
Consider the general problem of minimum-error discrimination involving $N$ mixed qubit states $\{p_i,\rho_i\}_{i=1}^N$. Building upon the previous discussions, we introduce the quantities $\gamma_0^{(i)}$, where $i$ corresponds to the number of non-null elements in a non-decomposable measurement. We define $\gamma_0^{(i)}$ as the maximum value within each family, representing the maximum guessing probability for each case. Thus, the overall guessing probability for the problem can be expressed as:
\begin{equation}\label{pguessNN}
P_{\textrm{guess}} = \gamma_0 = \max \{\gamma_0^{(i)}\}_{i=1}^4.
\end{equation}
\end{theorem}

\begin{proof}
The proof is straightforward. Just consider that for each $\gamma_0^{(i)}$ there is a POVM $\{\pi_j^{(i)}\}$ such that $\gamma_0^{(i)}=\sum_{j=1}^i p_j^i\Tr(\pi_j^i \rho_j^i)$ where $\{\rho_j^i\}_{j=1}^i$ are the states that contribute in obtaining the related $\gamma_0^{(i)}$. Therefore a $\gamma_0^{(j)}$ smaller than \eqref{pguessNN} cannot be an optimal answer as there will always be another POVM that yields a higher value for the guessing probability, contradicting optimality.
\end{proof}
Therefore, we have all we need to solve a general problem of $N$ qubit states.
\emph{To simplify this task and save time, we have developed a simple code using Mathematica. This code, based on our four-step instruction, is designed with the primary task of searching for the simplest non-decomposable answer, which includes $\boldsymbol{\gamma}$ and $P_{\textrm{guess}}$, as well as the detectable states for that answer \cite{RouhbakhshGithub2022}.}
\begin{figure}[t]
\mbox{\includegraphics[scale=0.44]{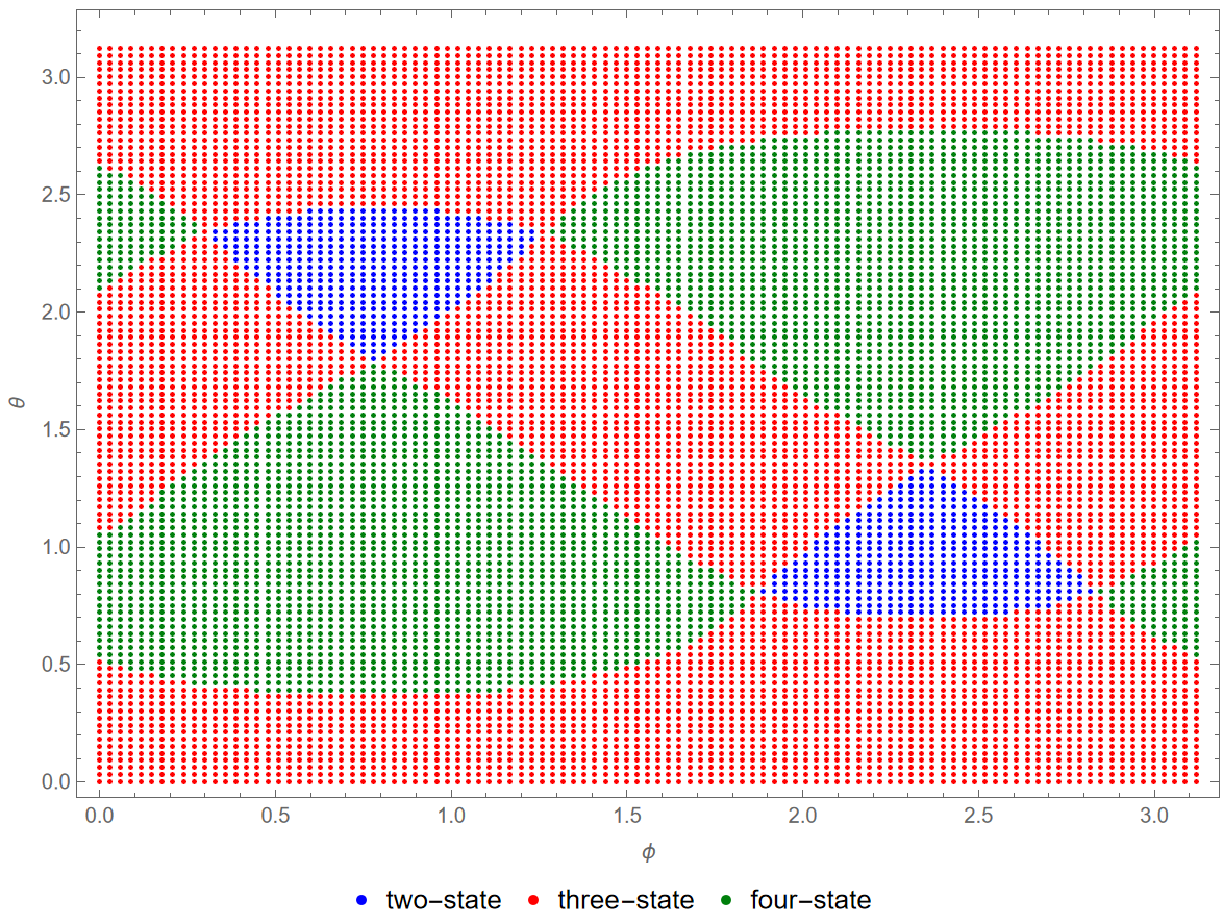}}
\caption{The simplest non-decomposable answers for five qubit states, all possessing the same purity $\zeta$. Four of them are equiprobable symmetric states ($p=0.195$), and one is a disturbing state ($p_5=\delta=0.22$), which rotates on a sphere with radius $\zeta$ and polar angles $\{\phi, \theta\}$. In this situation, the three-state case reveals $\boldsymbol{\gamma}$ in most areas.}
\label{fig:5states}
\end{figure}
\section{Examples} \label{section:examples}
To show the strength of our work to solve a general problem of $N$ mixed qubit states with arbitrary prior probabilities, we provide some examples using the Mathematica code. we first consider an example of a five-qubit state consisting of four equiprobable symmetric qubit states \eqref{foursymmetricpurestates} ($p_1=p_2=p_3=p_4=p$), and one more state $v_5=\zeta(\sin{\theta}\cos{\phi},\sin{\theta}\sin{\phi}, \cos{\theta})$ with a different priori probability $\delta=1-4p$. Our goal is to analyze how this state affects the optimal answer of the four symmetric states, for which we already know the optimal answer is a four-state case when $\delta$ is small. For $\delta\leq 0.2$, it can be observed that the only optimal answer is obtained through a four-element measurement. Thus, the optimal measurement remains the same as the problem of four equiprobable symmetric qubit states with the same purity, i.e., $\boldsymbol{\gamma}=0$, and is not disturbed by the new state. However, in this case, the guessing probability decreases as $\delta$ increases ($P_{\textrm{guess}}=2p\zeta=\zeta\frac{1-\delta}{2}$). When $\delta$ surpasses this threshold ($\delta > 0.2$), other cases will emerge. Fig.~\ref{fig:5states} illustrates the situation for $\delta=0.22$ and $p=0.195$. Based on the location of the state $\rho_5$ on the Bloch sphere, all three cases can happen. As $\delta$ increases, the region in which a four-state case is possible diminishes and eventually disappears. Consequently, $\boldsymbol{\gamma}$ can be determined using either a two-state or three-state case. Furthermore, for higher values of $\delta$, a two-state case will give the optimal answer.

To proceed, we first reassess all five states from the previous example. Then, we change the third and fifth states. The spherical coordinates of state $\rho_3$ are given by $\{3\pi/4, \arccos{(-1/\sqrt{3})}\}$. We replace its azimuthal angle with a variable $\omega$. Additionally, we can set the polar angle of $\rho_5$ to $\pi/3$. Their priori probabilities are unchanged as before. Consequently, there are three fixed states, with two rotating states on circles. The answers for $0<\phi, \omega<2\pi$ are shows Fig.~\ref{fig:5states2}.
\begin{figure}[t]
\mbox{\includegraphics[scale=0.44]{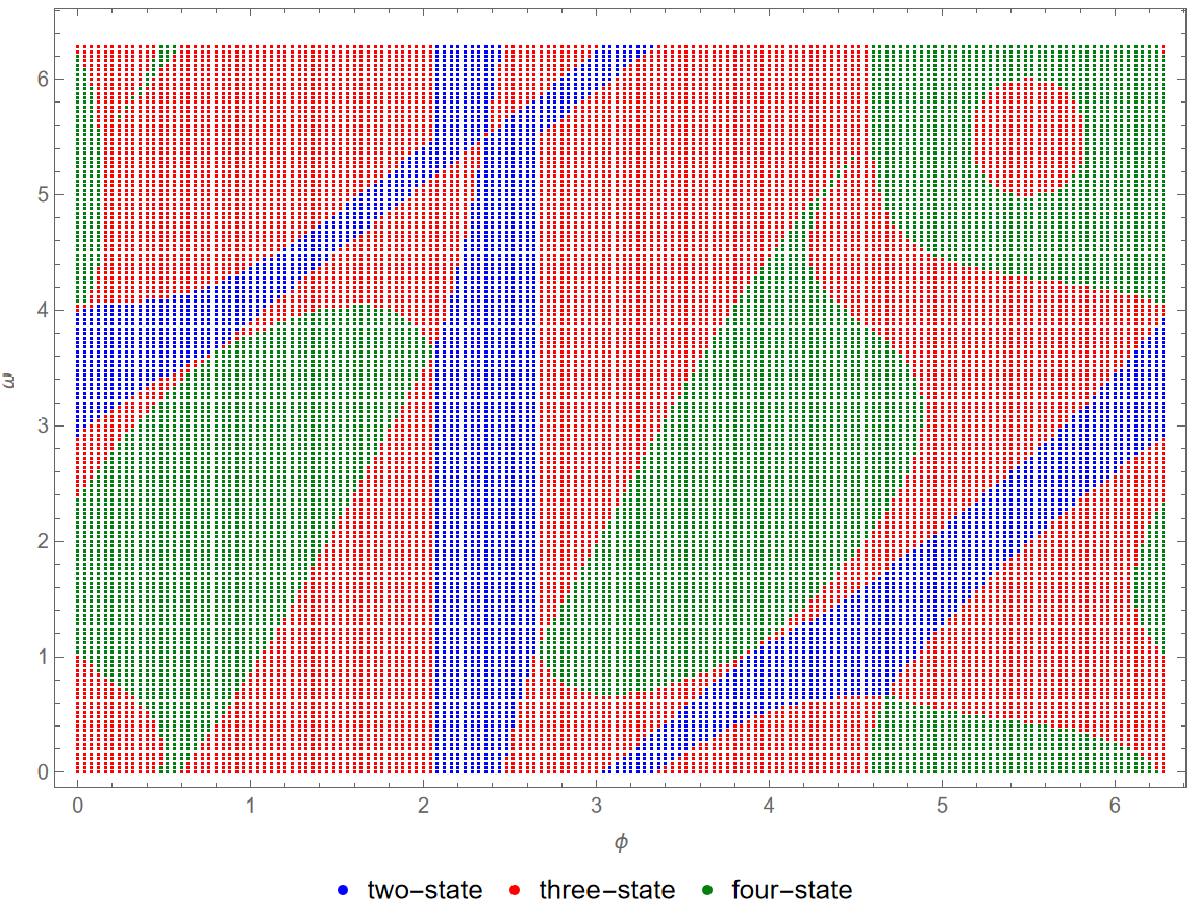}}
\caption{The simplest non-decomposable answers for five qubit states, which are defined by: $v_1=\zeta/\sqrt{3}(1, 1, 1)$, $v_2=\zeta/\sqrt{3}(1,- 1,- 1)$ , $v_3=\zeta/\sqrt{3}(\sin{a}\cos{\omega},\sin{a}\sin{\omega},\cos{a})$, $v_4=\zeta/\sqrt{3}(-1,- 1, 1)$, and $v_5=\zeta/2(\sqrt{3}\cos{\phi}, \sqrt{3}\sin{\phi}, 1)$ with $a=\arccos{(-1/\sqrt{3})}$. All states are mixed states with the same purity $\zeta$. Four states have the same priori probabilities $p_1=p_2=p_3=p_4=0.195$ (with $p_5=1-\sum_{i\neq 5} p_i=0.22$). We have two rotating states with different priori probabilities: 1- The state $\rho_3$ has a fixed polar angle $a$ and rotating azimuthal angle $\omega$. 1- The state $\rho_5$ has a fixed polar angle $\pi/3$ and rotating azimuthal angle $\phi$. In this case, the azimuthal angles change from zero to $2\pi$.}
\label{fig:5states2}
\end{figure}\\

For another example, we consider a more general case of six qubit states with no symmetries or equal priori probabilities. We denote the Bloch vectors of these states as
\begin{eqnarray}
\{p_1,\boldsymbol{v_1}\} &=& \{0.23, 0.85(\sin{\theta}\cos{\phi}, \sin{\theta}\sin{\phi}, \cos{\theta})\} \nonumber\\
\{p_2,\boldsymbol{v_2}\} &=& \{0.22, (-0.25, 0.64, 0.20)\} \nonumber\\
\{p_3,\boldsymbol{v_3}\} &=& \{0.20, (0.45, -0.50, 0.55)\} \nonumber\\
\{p_4,\boldsymbol{v_4}\} &=& \{0.19, (0.75, -0.25, -0.42)\} \nonumber\\
\{p_5,\boldsymbol{v_5}\} &=& \{0.10, (-0.50, 0.28, -0.57)\} \nonumber\\
\{p_6,\boldsymbol{v_6}\} &=& \{0.06, (0.33, 0.45, -0.82)\}
\label{sixarbitraryqubits}
\end{eqnarray}
\begin{figure}[t]
\mbox{\includegraphics[scale=0.51]{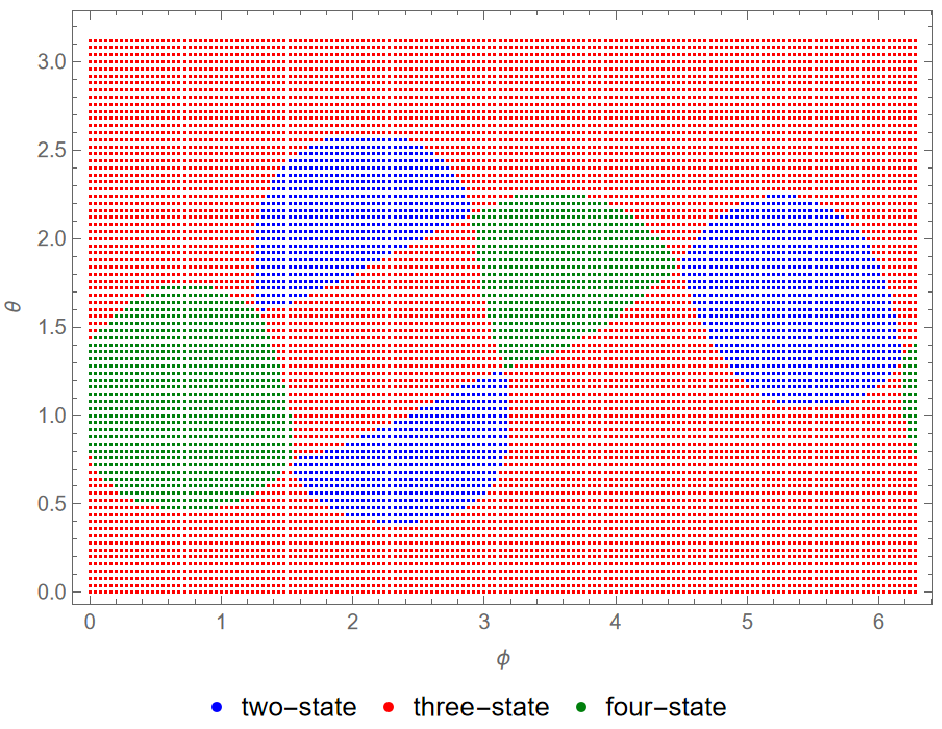}}
\caption{The simplest non-decomposable answers for six qubit states of Eq.~\eqref{sixarbitraryqubits}. \emph{No symmetry is considered!}. All the prior probabilities and purities of these states are different. In this case, the state $\rho_1$ is free to rotate on a sphere of radius 0.85 with spherical angles $\{\phi, \theta\}$.}
\label{fig:6states}
\end{figure}
The first qubit state is characterized by two spherical angles, $\phi$ and $\theta$, allowing it to freely rotate on a sphere with a radius of $0.85$. Figure~\ref{fig:6states}  illustrates the different regions where the $\boldsymbol{\gamma}$ can be obtained. Based on this figure, a three-element measurement will fulfill the guessing probability and there is no need to go to the next step to find $\boldsymbol{\gamma}$. Moreover, the regions for four-element measurements are relatively small. However, this does not imply the absence of optimal four-element measurements for the remaining regions where two-state and three-state solutions exist. For instance, for $\theta=1.585948$ and $\phi=1.288376$, there are three possible optimal measurements with common $\boldsymbol{\gamma}=(0.070525, 0.057738, 0.048073)$ and $P_{\textrm{guess}}=0.370574$. Two of these measurements, which are non-decomposable answers, can be obtained by considering two three-state sets: $\{\rho_1, \rho_2, \rho_3\}$ and $\{\rho_1, \rho_3, \rho_4\}$. The third one, which is decomposable, can be obtained using the four states: $\{\rho_1, \rho_2, \rho_3, \rho_4\}$. It should be emphasized that the optimal answer in the green area cannot be achieved using either two-element or three-element measurements.

As the last example, let us consider the special case of $N$ qubit states $\rho_i$ with equal a priori probabilities, i.e., $p_i=1/N$ for $i=1,\cdots, N$. It is shown in appendix \ref{Qubit states: Equal priori probabilities} that to achieve $P_{\textrm{guess}}$, we need to choose a sphere with a maximum number of states lying on. The sphere is referred to as the \emph{circumsphere}, which is simply the smallest possible sphere that encompasses all the $\boldsymbol{v}_i$'s (the same result was also obtained in \cite{BaePRA2013} using a different approach). The circumsphere is characterized by $\{R, \boldsymbol{O}\}$, where $R=|\boldsymbol{v}_i-\boldsymbol{O}|$ is its radius, and $\boldsymbol{O}$ represents its  circumcenter (Fig.~\ref{fig:circumsphere}). Considering that $0 < R \leq 1$ and $P_{\textrm{guess}}=\gamma_0$, we can derive
\begin{equation} \label{eq:defr}
\gamma_0=\frac{1}{N}\left( 1+R \right),
\end{equation}
and
\begin{equation} \label{eq:pguesslimits}
\frac{1}{N}<P_{\textrm{guess}}\leq\frac{2}{N}.
\end{equation}
Moreover, Eq.~\eqref{eq:ni} reduces to
\begin{equation} \label{eq:nieq}
\boldsymbol{\hat{n}_i}=\frac{\boldsymbol{v}_i-\boldsymbol{O}}{|\boldsymbol{v}_i-\boldsymbol{O}|}.
\end{equation}
A more detailed discussion on the case of equal a priori probabilities can be found in Appendix \ref{Qubit states: Equal priori probabilities}.

\begin{figure}
\mbox{\includegraphics[scale=0.35]{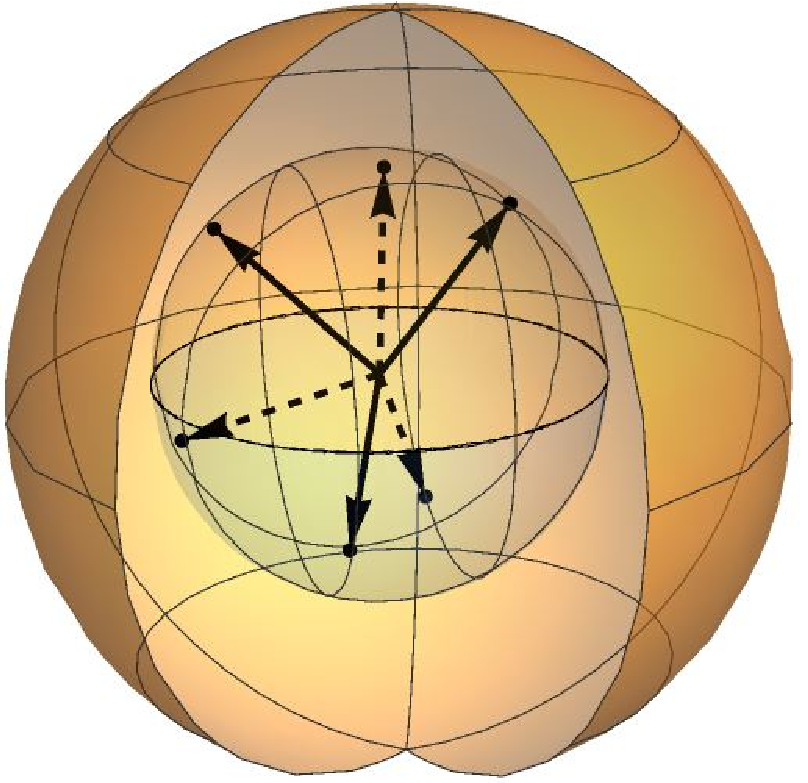}}
\caption{Geometrical representation of the circumsphere (interior sphere), which is inside the Bloch sphere (exterior sphere).}
\label{fig:circumsphere}
\end{figure}

\section{Conclusion} \label{section:conclusion}
In this paper, we have revisited the problem of minimum-error discrimination for mixed qubit states, aiming to obtain discrimination parameters in a constructive way by employing the necessary and sufficient Helstrom condition. Our approach is based on representing qubit states using Bloch vectors. For arbitrary priori probabilities, each pair of qubit states forms a hyperbola, and the desired $\boldsymbol{\gamma}$ lies on a part of this hyperbola close to the more probable state.

Through these tools, we introduced an instruction to find the Lagrange operator $\Gamma$, allowing us to identify all optimal POVM measurements. We also explored properties of the POVM answers, involving the geometry of the polytope of qubit states within the Bloch sphere, and introduced classes of answers, such as unchanged guessing probability and unchanged measurement operators.

We show that for an optimal strategy, there might be some states that cannot be detectable, meaning that their associated POVM elements are null. So, in the problem of ME discrimination of $N$ qubits $\{\rho_i\}_{i=1}^N$ some states might be undetectable. %We categorize these states as unguessable, nearly guessable, and guessable states, assuming that there are $M$ guessable states in a typical optimal problem ($1\leq M \leq N$).

Additionally, we demonstrated that every POVM set $M$ can be divided into a limited number of non-decomposable POVM subsets, $E'$. These subsets offer an alternative approach to constructing a general ME answer to the given problem, which may prove practical when certain states' detection is not a priority, particularly in scenarios where preparing measurement operators is expensive.

To illustrate our proposed instruction, we provided solutions for specific cases of $N\leq4$. For the two-state case, the Helstrom formula is already known. In the three-state case, we reduced the problem to that of three-qubit states on the $x-z$ plane, enabling a comprehensive analysis. Furthermore, we addressed the specific case of trine states with arbitrary priori probabilities, corroborating previous findings. Finally, we solved the four-state qubit case for the first time, utilizing the geometric properties of a tetrahedron and the intersections among $\binom{4}{2}=6$ hyperbolas derived from each pair of states.

Furthermore, we applied our instruction to solve examples for cases with $N\geq4$, including instances of five- and six-qubit states with non-equal priori probabilities and a general case of $N$ states with equal prior probabilities. In the latter case, finding $\Gamma$ corresponds to identifying a sphere with a maximum number of states on it.
\section*{Note on similar work}
\addcontentsline{toc}{section}{Note on Similar Work}
\textit{Subsequent to the submission of this paper, a separate study addressing a subsection of the problem discussed in our work (Section\ref{subsection:Four States}) was published \cite{HaQIP2023}, employing a distinctly different approach. It is noteworthy that both works were conducted independently.}
\acknowledgments
We would like to thank S. J. Akhtarshenas for many fruitful discussions. This work was supported by projects APVV-18-0518 (OPTIQUTE), VEGA 2/0161/19 (HOQIP). AG was further supported by funding from QuantERA, an ERA-Net cofund in Quantum Technologies, under the project eDICT.
\section*{Data availability} 
The data generated and analyzed during the current study are available on GitHub \cite{RouhbakhshGithub2022}, as well as from the corresponding author upon reasonable request.
\section*{Declaration}
\textbf{Conflict of interest}. The authors have no competing interests to declare that are relevant to the content of this article.

\appendix
\section{The rotation matrix for three qubit-state discrimination} \label{section:The Rotation Matrix}
Any qubit state $\rho_i$ can be identified with a point inside the Bloch sphere using its Bloch vector $\boldsymbol{v}_i$. The same is true for the multiplication of the state by its prior probability, i.e., $\tilde{\rho}_i=p_i\rho_i$ and $\boldsymbol{\tilde{v}_i}=p_i \boldsymbol{v}_i$. Three states with their corresponding priori probabilities, represent three points inside the Bloch sphere. From geometry, we know that three non-collinear points determine a plane, and each plane can be described by its normal vector, which is a vector orthogonal to the plane (i.e., orthogonal to every directional vector of the plane). By having the normal vector at hand, in the next step, we can find the rotation matrix that aligns this vector in the $y$-direction. This rotation matrix then rotates each $\boldsymbol{\tilde{v}_i}$ in a way that the y-component of all $\boldsymbol{\tilde{v}_i}$'s becomes equal. With a translation along the $y$-axis, we can eliminate the $y$-components of these rotated vectors. Finally, with an additional rotation in the $x-z$ plane, we can rotate $\boldsymbol{\tilde{v}_i}$'s in a way that the state with the largest priori probability is aligned in the $z$-direction. Since these rotations and the translation do not affect the relative distances and angles between states, the guessing probability for this new set of states is equal to the original one. The POVMs can also be easily related through a rotation, as explained in Sect.~\ref{subsection:Three States}.

To obtain corresponding rotation matrix, consider three points $P_1$,$P_2$ and $P_3$ in the Bloch sphere
\begin{equation}
P_1=
\begin{pmatrix}
x_1 \\
y_1 \\
z_1
\end{pmatrix} ,
P_2=
\begin{pmatrix}
x_2 \\
y_2 \\
z_2
\end{pmatrix} ,
P_3=
\begin{pmatrix}
x_3 \\
y_3 \\
z_3
\end{pmatrix} .
\end{equation}
The plane containing these three points is defined by
\begin{equation}
a x+b y+ c z + d =0  ,
\end{equation}
where the coefficients $a$, $b$ and $c$ determine the components of normal vector $\boldsymbol{n}$
\begin{equation}
\boldsymbol{n}=
\begin{pmatrix}
a \\
b \\
c
\end{pmatrix} .
\end{equation}
Since the normal vector $\boldsymbol{n}$ is the unit vector orthogonal to every directional vector of the plane, it can be obtained using the following equation
\begin{equation} \label{eq:normaln}
\boldsymbol{\hat{n}}=\frac{(P_2-P_1) \times (P_3-P_2)}{|(P_2-P_1) \times (P_3-P_2)|} .
\end{equation}

Our purpose is to rotate this plane such that the unit vector $\boldsymbol{\hat{n}}$ aligns with vector $\boldsymbol{\hat{j}}$
\begin{equation}
\boldsymbol{\hat{j}}=
\begin{pmatrix}
0 \\
1 \\
0
\end{pmatrix}.
\end{equation}
The desired rotation can be obtained by following this instruction:\\
Firstly, let us define the vector $\boldsymbol{V}$ and matrix $A$ as
\begin{equation}
\boldsymbol{V}=\boldsymbol{\hat{n}} \times \boldsymbol{\hat{j}},
\end{equation}
and
\begin{equation}
A=
\begin{pmatrix}
0 & - V_z & 0 \\
V_z & 0 & - V_x \\
0 & V_x & 0
\end{pmatrix}
\end{equation}
Then the rotation Matrix can be written as
\begin{equation}
R= \mathbb{1} + A + (\frac{1-c}{s^2})A^2 ,
\end{equation}
where the coefficients $c$ and $s$ are
\begin{eqnarray}
s&=&||V||=\sqrt{V_{x}^2+V_{y}^2+V_{z}^2}, \nonumber\\
c&=&\boldsymbol{\hat{n}}\cdot\boldsymbol{\hat{j}}. \nonumber
\end{eqnarray}
A more straightforward approach for our specific problem is using the following formulation for the rotation matrix $R$
\begin{equation}
R=
\begin{pmatrix}
1-\frac{n_x^2}{1+n_y} & -n_x & -\frac{n_x n_y}{1+n_y} \\
n_x & n_y & n_z \\
-\frac{n_x n_y}{1+n_y} & -n_z & 1-\frac{n_x^2}{1+n_z}
\end{pmatrix}.
\end{equation}
where $\boldsymbol{\hat{n}}=(n_x, n_y, n_z)$ is defined in Eq.~\eqref{eq:normaln}.\\

\section{Tetrahedron geometry for the problem of four-state qubit discrimination} \label{section:Tetrahedron geometry}
In this section, we briefly review the material needed for solving the problem of four qubit states discrimination in Sect.~(\ref{subsection:Four States}). As previously mentioned, four qubits that are not on a plane form a tetrahedron in the Bloch sphere. A tetrahedron consists of four triangular faces, six straight edges, and four vertices. There are three types of angles for a tetrahedron: 12 face angles which are the regular angles of each triangle; 6 dihedral angles associated with each edge of the tetrahedron, which represent the angles between each pair of connected faces; and 4 solid angles for each vertex of the tetrahedron. For the problem of four qubit states, to make these quantities clearer, let us define the following quantities for the lengths of edges and dihedral angles:\\
$\tilde{d}_{ij}$: the length of the edge connecting two qubits $\boldsymbol{\tilde{v}_{i}}$ and $\boldsymbol{\tilde{v}_{j}}$ in the Bloch sphere. \\
$\theta_{ij}$: the dihedral angle between two adjacent faces connected by the edge $ij$. \\
By using these quantities, a relation can be established for the dihedral angle $\theta_{ij}$ \cite{WirthJMC2014} 
\begin{equation}
\cos(\theta_{ij})=\frac{D_{ij}}{\sqrt{D_{ijk}{D_{ijl}}}},
\label{dihedralangle}
\end{equation}
where 
\begin{widetext}
\begin{eqnarray}
D_{ij}&=&-\tilde{d}_{ij}^4+(\tilde{d}_{ik}^2+\tilde{d}_{il}^2+\tilde{d}_{jk}^2+\tilde{d}_{jl}^2-2\tilde{d}_{kl}^2)\tilde{d}_{ij}^2 
+(\tilde{d}_{ik}^2-\tilde{d}_{jk}^2)(\tilde{d}_{jl}^2-\tilde{d}_{il}^2), \nonumber \\
D_{ijk}&=& 
-(\tilde{d}_{ij}+\tilde{d}_{ik}+\tilde{d}_{jk})(\tilde{d}_{ij}+\tilde{d}_{ik}-\tilde{d}_{jk})(\tilde{d}_{jk}+\tilde{d}_{ij}-\tilde{d}_{ik})(\tilde{d}_{ik}+\tilde{d}_{jk}-\tilde{d}_{ij}).
\end{eqnarray}
\end{widetext}
To find the guessing probability for qubit states when all states are detectable in an optimal strategy, $\boldsymbol{\gamma}$ must be an interior point within the tetrahedron. By connecting it with straight lines to each vertex, we obtain five tetrahedrons. Let us consider two tetrahedrons: $T_1=\boldsymbol{\gamma}\boldsymbol{\tilde{v}_{1}}\boldsymbol{\tilde{v}_{2}}\boldsymbol{\tilde{v}_{3}}$ and $T_2=\boldsymbol{\gamma}\boldsymbol{\tilde{v}_{1}}\boldsymbol{\tilde{v}_{2}}\boldsymbol{\tilde{v}_{4}}$. The dihedral angles $\theta_{12}^{T_1}$ and $\theta_{12}^{T_2}$ correspond to the common edge $\boldsymbol{\tilde{v}_{1}}\boldsymbol{\tilde{v}_{2}}$ in these two tetrahedrons. Therefore, there exists a simple relation between these angles and the dihedral angle $\alpha$ for the edge $\boldsymbol{\tilde{v}_1}\boldsymbol{\tilde{v}_{2}}$ of the main tetrahedron $\boldsymbol{\tilde{v}_{1}}\boldsymbol{\tilde{v}_{2}}\boldsymbol{\tilde{v}_{3}}\boldsymbol{\tilde{v}_{4}}$.
\begin{equation}
\alpha=\theta_{12}^{T_1}+\theta_{12}^{T_2}.
\end{equation}
Using \eqref{dihedralangle} for tetrahedron $T_1$, we have
\begin{eqnarray}
\cos(\theta_{12}^{T_1})&=&\frac{D_{12}^{T_1}}{\sqrt{D_{123}^{T_1}{D_{12\gamma}^{T_1}}}} \nonumber% \\
\end{eqnarray}
where $D_{12}^{T_1}=l_0+l_1|\boldsymbol{\gamma}-\boldsymbol{\tilde{v}_1}|$. The parameters $l_0$ and $l_1$ were defined in \eqref{l0l1}. After performing some algebraic calculations, we end up with the following equation
\small
\begin{widetext}
\begin{eqnarray}
|\boldsymbol{\gamma}-\boldsymbol{\tilde{v}_1}|=\enspace\enspace\enspace\enspace \enspace\enspace\enspace\enspace\enspace\enspace\enspace\enspace \enspace\enspace\enspace\enspace\enspace\enspace\enspace\enspace \enspace\enspace\enspace\enspace\enspace\enspace\enspace\enspace \enspace\enspace\enspace\enspace\enspace\enspace\enspace\enspace \enspace\enspace\enspace\enspace\enspace\enspace\enspace\enspace \enspace\enspace\enspace\enspace\enspace\enspace\enspace\enspace \enspace\enspace\enspace\enspace\enspace\enspace\enspace\enspace \enspace\enspace\enspace\enspace\enspace\enspace\enspace\enspace \enspace\enspace\enspace\enspace\enspace\enspace\enspace\enspace \enspace\enspace\enspace\enspace\enspace\enspace\enspace\enspace \enspace\enspace\enspace\enspace\enspace\enspace\enspace\enspace \enspace\enspace\enspace\enspace\enspace\enspace\enspace\enspace \enspace\enspace\enspace\enspace\enspace\enspace\enspace\enspace \enspace\enspace\enspace\enspace \\
\frac{-2D_{123}p_{12}(\tilde{d}_{12}^2-p_{12}^2) \cos^2(\theta_{12}^{T_1})-l_0 l_1-\cos(\theta_{12}^{T_1})\sqrt{\bigl(4\cos^2(\theta_{12}^{T_1})D_{123}(\tilde{d}_{12}^2-p_{12}^2)\tilde{d}_{12}^2+l_1^2(\tilde{d}_{12}^2-p_{12}^2)+4p_{12}l_0 l_1-4 l_0^2\bigr)D_{123}(\tilde{d}_{12}^2-p_{12}^2)}}{4 D_{123}(\tilde{d}_{12}^2-p_{12}^2) \cos^2(\theta_{12}^{T_1})+l_1^2}, \enspace\enspace\enspace\enspace \enspace\enspace\enspace\enspace  , \nonumber
\end{eqnarray}
Similarly, for tetrahedron $T_2$ one can write
\begin{eqnarray}
|\boldsymbol{\gamma}-\boldsymbol{\tilde{v}_1}|=\enspace\enspace\enspace\enspace \enspace\enspace\enspace\enspace\enspace\enspace\enspace\enspace \enspace\enspace\enspace\enspace\enspace\enspace\enspace\enspace \enspace\enspace\enspace\enspace\enspace\enspace\enspace\enspace \enspace\enspace\enspace\enspace\enspace\enspace\enspace\enspace \enspace\enspace\enspace\enspace\enspace\enspace\enspace\enspace \enspace\enspace\enspace\enspace\enspace\enspace\enspace\enspace \enspace\enspace\enspace\enspace\enspace\enspace\enspace\enspace \enspace\enspace\enspace\enspace\enspace\enspace\enspace\enspace \enspace\enspace\enspace\enspace\enspace\enspace\enspace\enspace \enspace\enspace\enspace\enspace\enspace\enspace\enspace\enspace \enspace\enspace\enspace\enspace\enspace\enspace\enspace\enspace \enspace\enspace\enspace\enspace\enspace\enspace\enspace\enspace \enspace\enspace\enspace\enspace\enspace\enspace\enspace\enspace \enspace\enspace\enspace\enspace \\
\frac{-2D_{124}p_{12}(\tilde{d}_{12}^2-p_{12}^2) \cos^2(\theta_{12}^{T_2})-l'_0 l'_1-\cos(\theta_{12}^{T_2})\sqrt{\bigl(4\cos^2(\theta)D_{124}(\tilde{d}_{12}^2-p_{12}^2)\tilde{d}_{12}^2+{l'}_1^2(\tilde{d}_{12}^2-p_{12}^2)+4p_{12}l'_0 l'_1-4 {l'}_0^2\bigr)D_{124}(\tilde{d}_{12}^2-p_{12}^2)}}{4 D_{124}(\tilde{d}_{12}^2-p_{12}^2) \cos^2(\theta_{12}^{T_2})+{l'}_1^2}, \enspace\enspace\enspace\enspace \enspace\enspace\enspace\enspace,  \nonumber
\end{eqnarray}
\end{widetext}
\normalsize
where $l'_0$ and $l'_1$ can be obtained from $l_0$ and $l_1$ by replacing $3 \leftrightarrow 4$. Therefore, by equating these two equations and using the relation $\theta_{12}^{T_2}=\alpha-\theta_{12}^{T_1}$, the angle $\theta_{12}^{T_1}$ can be obtained. For example, in the problem of four symmetric qubit states in section \ref{subsection:Four States}, we can easily obtain that $\theta_{12}^{T_1}=\theta_{12}^{T_2}=\frac{\alpha}{2}$.

\section{Qubit states: Equal priori probabilities} \label{Qubit states: equal priori probabilities}
In the case of equal priori probabilities, Eqs. \eqref{eq:gamma0-p}  and \eqref{eq:gamma0-pgeq}  reduce to

\begin{equation} \label{eq:gamma0-1/N}
|\boldsymbol{v}_i-\boldsymbol{O}|=N\gamma_0-1,
\end{equation}
for $i=1,\cdots,M$ and
\begin{equation} \label{eq:gamma0-1/N2}
|\boldsymbol{v}_i-\boldsymbol{O}|< N\gamma_0-1,
\end{equation}
for $i=M+1,\cdots,N$, respectively,  where $\boldsymbol{O}=N\boldsymbol{\gamma}$.

These equations admit geometrical interpretation; Eq. \eqref{eq:gamma0-1/N} defines a sphere with radius $R=N\gamma_0-1$ centered at the point $\boldsymbol{O}$. All Bloch vectors, $\boldsymbol{v}_i$'s, are embedded within this sphere. Specifically, $M$ Bloch vectors $\boldsymbol{v}_i$'s with $i=1,\cdots,M$ lie on the surface of the sphere, while the remaining $N-M$ Bloch vectors $\boldsymbol{v}_i$'s with $i=M+1,\cdots,N$ are located inside the sphere (see Fig.~\ref{fig:circumsphere}). As we mentioned previously, the measurement operators associated with the former are given by Eq.~\eqref{eq:defpi}. However, for the latter Bloch vectors, we have $\pi_i=0$, indicating that their corresponding states do not appear to have any output in the discrimination process.

\begin{lemma}
In the optimal measurement, the number of states satisfying  Eq.~\eqref{eq:gamma0-1/N} is maximum.
\end{lemma}
\begin{proof}
We suppose that at most $M$ states can satisfy Eq.~\eqref{eq:gamma0-1/N}. Without loss of generality, we label these states from $1$ to $M$. Then, the guessing probability can be written as
\begin{equation}
P_{\textrm{guess}}^{M}=\max_{\pi_1,\cdots,\pi_M}  \sum\limits_{i=1}^M \Tr(\tilde{\rho}_i\pi_{i}) , \enspace\enspace\enspace\enspace \sum_{i=1}^M \pi_i=\mathbb{1}.
\end{equation}
Now, we consider a second strategy where at most $M-1$ states satisfy Eq.~\eqref{eq:gamma0-1/N}. Denoting the corresponding measurement operators by $\pi^\prime_i$ ($i=1,\cdots,M-1$), the guessing probability is obtained by
\begin{equation}
P_{\textrm{guess}}^{M-1}=\max_{\pi_1^\prime,\cdots,\pi_{M-1}^\prime}  \sum\limits_{i=1}^{M-1} \Tr(\tilde{\rho}_i\pi_{i}^\prime), \enspace\enspace\enspace\enspace \sum_{i=1}^{M-1} \pi_i^\prime=\mathbb{1}.
\end{equation}
Clearly, $P_{\textrm{guess}}^{M-1}\leq P_{\textrm{guess}}^{M}$ since $P_{\textrm{guess}}^{M-1}$ can be obtained from $P_{\textrm{guess}}^{M}$ by introducing an additional constraint $\pi_M=0$.
\end{proof}
The lemma implies that to achieve $P_{\textrm{guess}}$, we have to choose a sphere that contains the maximum number of states lying on. It is simply \emph{the minimal sphere covering all the $\boldsymbol{v}_i$}'s (the same result was obtained in \cite{BaePRA2013,BaeNJP2013} with a different approach). We refer to it as the \emph{circumsphere}, which is defined  by $\{R, \boldsymbol{O}\}$, where $R=|\boldsymbol{v}_i-\boldsymbol{O}|$ represents its radius, and $\boldsymbol{O}$ denotes its  circumcenter. Based on this, considering the fact that $0 < R \leq 1$ and $P_{\textrm{guess}}=\gamma_0$, we can deduce from Eq.~\eqref{eq:gamma0-1/N}
\begin{equation} \label{eq:defr}
\gamma_0=\frac{1}{N}\left( 1+R \right),
\end{equation}
and
\begin{equation} \label{eq:pguesslimits}
\frac{1}{N}<P_{\textrm{guess}}\leq\frac{2}{N}.
\end{equation}
Moreover, Eq.~\eqref{eq:ni} reduces to
\begin{equation} \label{eq:nieq}
\boldsymbol{\hat{n}_i}=\frac{\boldsymbol{v}_i-\boldsymbol{O}}{|\boldsymbol{v}_i-\boldsymbol{O}|}.
\end{equation}
Note that, \emph{as long as the circumsphere remains fixed, it is possible to modify the angles between $N$ qubit states while keeping the guessing probability unchanged.}\\
The following corollaries can be immediately derived from the above results.
\begin{corollary}\label{corollary:7}
We suppose that there are $N$ qubit states with equal priori probabilities, $\{\frac{1}{N}, \rho_i \}_{i=1}^N$, covering with the circumsphere $\{R, \boldsymbol{O}\}$. In this case, Eq.~\eqref{eq:sigmaalpha} holds with some known $\boldsymbol{\hat{n}_i}$'s from Eq.~\eqref{eq:nieq} (note that only the states lying on the circumsphere are defined by $\boldsymbol{\hat{n}_i}$). Then\\ (i) From the previous discussions, if we add $K$ different qubit states to the circumsphere, where all states have equal prior probabilities of $\frac{1}{N+K}$, the circumsphere remains invariant. (ii) It follows from Eq.~\eqref{eq:defpi} that all the optimal measurement operators for the original set $\{\frac{1}{N}, \rho_i \}_{i=1}^N$ are still optimal for the new set $\{\frac{1}{N+K}, \rho_i \}_{i=1}^{N+K}$ because the defined $\boldsymbol{\hat{n}_i}$'s from the original problem with $N$ states remain unchanged for the new problem, as given by Eq.~\eqref{eq:nieq}. However, this is not the only solution to the new problem. (iii) The guessing probability for the new problem can be expressed in terms of the guessing probability  $P_{guess}=\frac{1}{N}\left( 1+R \right)$ of the original one as $P_{\textrm{guess}}^{New}=\frac{N}{N+K} P_{\textrm{guess}}$, according to Eq.~\eqref{eq:defr}.
\end{corollary}

\begin{corollary}\label{corollary:8}
Consider the case of $N$ equiprobable qubit states with the same purity $|\boldsymbol{v}_i|$ that are located on the circumsphere $\{R, \boldsymbol{O}\}$. If any of the following statements is true, then the other statements will also be true. In other words, these statements can be used interchangeably in this particular case.
\begin{enumerate}[(i)]
\item
The circumcenter is at the origin, i.e. $\boldsymbol{O}=0$.
\item
The radius of the circumsphere is given by $R=|\boldsymbol{v}_i|$.
\item
The guessing probability is $P_{\textrm{guess}}=\frac{1}{N}\left( 1+|\boldsymbol{v}_i| \right)$.
\item The optimal measurement operators are given by Eq.~\eqref{eq:defpi} with $\boldsymbol{\hat{n}_i}=\frac{\boldsymbol{v}_i}{|\boldsymbol{v}_i|}$.
\end{enumerate}
\end{corollary}

All implications are trivial and can be inferred from the results given above.
In particular, we consider the case of $N$ equiprobable pure states, $|\boldsymbol{v}_i|=1$. From the corollary, it follows that the guessing probability of $N$ equiprobable pure qubit states with $\boldsymbol{O}$ at the origin is an example that achieves its maximum value $P_{\textrm{guess}}=2/N$. This is because the radius of this circumsphere cannot be greater than one.

\end{document}